\documentclass[a4paper,UKenglish,cleveref, autoref, thm-restate]{lipics-v2021}
\hideLIPIcs  %uncomment to remove references to LIPIcs series (logo, DOI, ...), e.g. when preparing a pre-final version to be uploaded to arXiv or another public repository

\usepackage{amsmath,amsthm,amssymb,appendix,bm,graphicx,hyperref,mathrsfs}
\newtheorem*{theorem*}{Theorem}
\usepackage{tikz}
\usetikzlibrary{positioning}

\newcommand{\bi}[2]{\neq_2^{#1,#2}}

\newcommand{\hol}[0]{\textsf{Holant}}
\newcommand{\eo}[0]{\textsf{EO}}
\newcommand{\csp}[0]{\textsf{CSP}}

\newcommand{\eoe}[0]{\text{HW}^=}
\newcommand{\eog}[0]{\text{HW}^\geq}
\newcommand{\eol}[0]{\text{HW}^\leq}
\newcommand{\eosg}[0]{\text{HW}^>}

\newcommand{\hw}[1]{\text{HW}^{#1}}

\newcommand{\su}[0]{\text{supp}}

\newcommand{\pin}[0]{\Delta}

\newcommand{\eom}[1][\text{M}]{\textsf{EO}^{#1}}

\newcommand{\ba}[1][0]{{{#1}-rebalancing}}
\title{P-time algorithms for typical \#$\eo$ problems} %TODO Please add

\author{Boning Meng\footnote{First authors.}}{Key Laboratory of System Software (Chinese Academy of Sciences) and State Key Laboratory of Computer Science, Institute of Software, Chinese Academy of Sciences; University of Chinese Academy of Sciences, Beijing, China}{mengbn@ios.ac.cn}{https://orcid.org/0009-0006-0088-1639}{}
\author{Juqiu Wang\footnotemark[1]}{Key Laboratory of System Software (Chinese Academy of Sciences) and State Key Laboratory of Computer Science, Institute of Software, Chinese Academy of Sciences; University of Chinese Academy of Sciences, Beijing, China}{wangjq21@ios.ac.cn}{https://orcid.org/0000-0001-9801-271X}{}
\author{Mingji Xia\footnotemark[1]}{Key Laboratory of System Software (Chinese Academy of Sciences) and State Key Laboratory of Computer Science, Institute of Software, Chinese Academy of Sciences; University of Chinese Academy of Sciences, Beijing, China}{mingji@ios.ac.cn}{https://orcid.org/0000-0002-3868-9910}{}
\authorrunning{B. Meng, J. Wang, M, Xia} %TODO mandatory. First: Use abbreviated first/middle names. Second (only in severe cases): Use first author plus 'et al.'
\Copyright{Boning Meng, Juqiu Wang and Mingji Xia} %TODO mandatory, please use full first names. LIPIcs license is "CC-BY";  http://creativecommons.org/licenses/by/3.0/

%TODO mandatory: Please choose ACM 2012 classifications from https://dl.acm.org/ccs/ccs_flat.cfm 

 %TODO mandatory; please add comma-separated list of keywords

\category{} %optional, e.g. invited paper

\relatedversion{} %optional, e.g. full version hosted on arXiv, HAL, or other respository/website
\funding{All authors are supported by National Key R\&D Program of China (2023YFA1009500), NSFC 61932002 and NSFC 62272448.}
\acknowledgements{
We sincerely thank Yicheng Pan for providing a simplified proof of Lemma \ref{lem:alluphasdelta1}.}%optional

\nolinenumbers %uncomment to disable line numbering

\begin{document}

\maketitle

\begin{abstract}
In this article, we study the computational complexity of counting weighted Eulerian orientations, denoted as \#\textsf{EO}. This problem is considered a pivotal scenario in the complexity classification for \textsf{Holant}, a counting framework of great significance. Our results consist of three parts. First, we prove a complexity dichotomy theorem for \#\textsf{EO} defined by a set of binary and quaternary signatures, which generalizes the previous dichotomy for the six-vertex model. Second, we prove a dichotomy for \#\textsf{EO} defined by a set of so-called pure signatures, which possess the closure property under gadget construction. Finally, we present a polynomial-time algorithm for \#\textsf{EO} defined by specific rebalancing signatures, which extends the algorithm for pure signatures to a broader range of problems, including \#\textsf{EO} defined by non-pure signatures such as $f_{40}$. We also construct a signature $f_{56}$ that is not rebalancing, and whether $\#\textsf{EO}(f_{56})$ is computable in polynomial time remains open.  
\end{abstract}

\section{Introduction}
In fields such as statistical physics, economics, machine learning, and combinatorics, the role of counting problems is becoming increasingly significant. Three well-founded frameworks have been put forth for the study of the complexity of counting problems: \#$\textsf{GH}$, $\#\csp$, and $\hol$. These frameworks are capable of expressing a wide range of natural counting problems and are of great significance in counting complexity. For example, counting vertex covers in a graph can be expressed in \#$\textsf{GH}$, while counting perfect matchings is a $\hol$ problem. The definitions of these frameworks are as follows; see \cite[Section 1.2]{cai2017complexity} and \cite[section 2.1]{cai2020beyond} for details. In this paper, we always restrict the variables in these counting problems to the Boolean domain, where each variable can only take values in $\{0,1\}$, by default.  

We begin by introducing some basic concepts. A \textit{Boolean variable} takes values from a Boolean domain consisting of two symbols $\{0,1\}$. Sometimes we treat it as $\mathbb{F}_2$, the finite field of size 2, to characterize special tractable classes.

A \textit{signature} $f$ with $r$ variables is a mapping from $\{0,1\}^r$ to $\mathbb{C}$. The value of $f$ on an input $\alpha$ is denoted as $f_{\alpha}$ or $f(\alpha)$. The set of all variables of $f$ is denoted by $\text{Var}(f)$, and its size by $\text{arity}(f)$. 
\begin{definition}[\#$\textsf{GH}$]
    A counting weighted graph homomorphisms problem $\#\textsf{GH}(\mathcal{F})$ parameterized by a binary signature $f$ is as the following: The input is a directed graph\footnote{In this article, graphs refer to multigraphs. It is always permissible for self-loops and parallel edges to be present.} $G = (V,E)$. The output is 

$$Z_G=\sum_{\sigma:V\to \{0,1\}}\prod_{(u,v)\in E} f(\sigma(u),\sigma(v))$$
\label{def:GH}
\end{definition}

\begin{definition}[$\#\csp$]
    A counting constraint satisfaction problem $\#\csp(\mathcal{F})$ parameterized by a set $\mathcal{F}$ is as the following: The input is an instance of $\#\csp(\mathcal{F})$, which consists of a finite set of variables $\{x_1,x_2,\cdots,x_n\}$ and a finite set $C$ of clauses. Each clause in $C$ contains a signature $f\in \mathcal{F}$ of arity $k$ and a sequence of variables of length $k$ $(x_{i_1},x_{i_2},\cdots,x_{i_k})$ from $\{x_1,x_2,\cdots,x_n\}$ \footnote{A variable can appear in the sequence more than one time.}. The output is 

$$\sum_{x_1,x_2,\cdots,x_n\in\{0,1\}}\prod_{(f,x_{i_1},x_{i_2},\cdots,x_{i_k})\in C}f(x_{i_1},x_{i_2},\cdots,x_{i_k})$$
\label{def:CSP}
\end{definition}

\begin{definition}[$\hol$]
    A $\hol$ problem $\hol(\mathcal{F})$ parameterized by a set $\mathcal{F}$ is as the following: The input is an signature grid $\Omega(G,\pi)$ over $\mathcal{F}$, denoted by $I$. Here, $G$ is a graph and $\pi:V\to \mathcal{F}\times S(E(v))$ assigns a signature $f_v\in\mathcal{F}$ of arity $|E(v)|$ to $v$ and a linear order to $E(v)$ for each $v\in V(G)$, where $E(v)$ are the edges adjacent to $v$. The output is the partition function of $\Omega$,

$$\text{Z}(I)=\hol_\Omega=\sum_{\sigma:E(G)\to\{0,1\}}\prod_{v\in V(G)}f_v(\sigma|_{E(v)})$$

The bipartite $\hol$ problem $\hol(\mathcal{F}\mid\mathcal G)$ is a $\hol$ problem over the signature grid $(G,\pi)$, where $G=(U,V,E)$ is a bipartite graph, and $\pi$ assigns signatures from $\mathcal F$ to vertices in $U$ and signatures from $\mathcal{G}$ to vertices in $V$.
\label{def:holant}
\end{definition}

The framework of $\hol$ is capable of expressing counting weighted graph homomorphisms (\#$\textsf{GH}$) and counting constraint satisfaction problems ($\#\csp$). Consequently, $\hol$ is widely regarded as one of the most general and significant frameworks in counting complexity.

Significant progress has been made in the study of \#$\textsf{GH}$ \cite{lovasz1967operations,hell1990complexity,dyer2000complexity,bulatov2005complexity,dyer2007counting,goldberg2010complexity,cai2013graph,cai2019decidable}, $\#\csp$ \cite{bulatov2005complexity,bulatov2007towards,dyer2009complexity,bulatov2009complexity,cai2009holant,dyer2010complexity,guo2011complexity,bulatov2012complexity,bulatov2013complexity,dyer2013effective,cai2014complexity,cai2016nonnegative,cai2017godel}, $\#\eo$ \cite{cai2018complexity,cai2020beyond,shao2024eulerian}, and $\hol$ \cite{cai2009holant,cai2011computational,cai2011holant*,cai2013vanishing,backens2017holant+,cai2018realholantc,backens2021full,lin2018holantnonnegative,cai2020holantoddarity,shao2020realholant}. Moreover, the computational complexity of \#$\textsf{GH}$ and $\#\csp$ has been fully characterized by dichotomy theorems \cite{cai2013graph,cai2014complexity}. In contrast, the complexity classification for $\hol$ remains unresolved.

 There have been several constructive attempts on $\hol$. A signature is said to be symmetric if all its values for inputs with the same number of 1's are identical. A complete dichotomy is established when all signatures are symmetric \cite{cai2013vanishing}. When signatures are not restricted to be symmetric, a complete dichotomy has been proven for nonnegative-weighted signatures \cite{lin2018holantnonnegative} and real-weighted signatures \cite{shao2020realholant}. Additionally, dichotomies exist for several special forms of complex-weighted $\hol$, such as $\hol^*$ \cite{cai2011holant*}, $\hol^+$ \cite{backens2017holant+}, and $\hol^c$ \cite{backens2021full}, with some given auxiliary signatures.

Recently, counting weighted Eulerian orientation problems ($\#\eo$) have also attracted researchers' attention, as real-weighted $\hol$ can express $\#\eo$ problems defined by signatures with the ARS property \cite{cai2020beyond} under a holographic transformation by $Z=\frac{1}{\sqrt{2}}\left(\begin{matrix}
1 &  1\\
\mathfrak{i} & -\mathfrak{i} 
\end{matrix}\right)$, while complex-weighted $\hol$ can express $\#\eo$ problems under the same transformation. The complexity classification for complex-weighted $\#\eo$ problems remains open, and therefore this article seeks a more generalized characterization of $\#\eo$ complexity.

\subsection{Counting weighted Eulerian orientation}
We first introduce the concept of $\#\eo$. We use $\hw{<}$ to denote the set of strings with fewer 1's than 0's. For example, $00110\in \hw{<}$ since the number of 1's (which is 2) is strictly less than the number of 0's (which is 3). We similarly define $\hw{=}$, $\hw{\leq}$, $\hw{>}$, and $\hw{\geq}$. 
The \textit{support} of a signature $f$, denoted $\su(f)$, is the set of all inputs for which $f$ is non-zero.  If $\su(f)\subseteq\hw{=}$, we call $f$ an $\eo$ signature. 

For any Eulerian graph $G$, let $\eo(G)$ denote all Eulerian orientations of $G$. For a given orientation in $\eo(G)$, we assign $0$ to the head and $1$ to the tail of each edge. Therefore, a Eulerian orientation corresponds to an assignment to both ends of every edge, where for each vertex, the number of adjacent ends assigned $0$ equals those assigned $1$.

\begin{definition}[{$\#\eo$}]
A $\#\eo$ problem $\#\eo(\mathcal{F})$ parameterized by a signature set $\mathcal{F}$ of $\eo$ signatures is as the following: The input is an $\eo$-signature grid $\Omega(G,\pi)$ over $\mathcal{F}$, denoted by $I$. Here, $G$ is a Eulerian graph and $\pi:V\to \mathcal{F}\times S(E(v))$ assigns a signature $f_v\in\mathcal{F}$ of arity $|E(v)|$ to $v$ and a linear order to $E(v)$ for each $v\in V(G)$, where $E(v)$ are the edges adjacent to $v$. The output is the partition function of $\Omega$,

$$\text{Z}(I)=\#\eo_\Omega=\sum_{\sigma\in\eo(G)}\prod_{v\in V}f_v(\sigma|_{E(v)}).$$
\label{defeo}
\end{definition}

Several complexity results have been established for $\#\eo$. In \cite{cai2018complexity}, a complexity dichotomy was proved for a single complex-weighted quaternary signature, known as the six-vertex model dichotomy. This result was later extended to planar graphs in \cite{cai2021new}. Furthermore, \cite{cai2020beyond} established a complexity dichotomy for signatures satisfying the arrow reversal symmetry (ARS) property, a setting commonly assumed in physics.

There are two motivations for studying $\#\eo$ problems. First, specific forms of $\#\eo$ problems appear in many areas such as statistical physics and combinatorics. In statistical physics, both the ice-type model and the six-vertex model \cite{pauling1935structure} are special cases of $\#\eo$ problems. The latter model, in particular, has emerged as a prominent focus within statistical physics and corresponds exactly to $\#\eo$ defined on 4-regular graphs. In combinatorics, the resolution of the Alternating Sign Matrix conjecture \cite{bressoud1999proofs} and the evaluation of the Tutte polynomial at (3,3) \cite{las1988evaluation} both relate to $\#\eo$ problems defined on specific graphs with particular signatures.

Second, we conjecture that resolving $\#\eo$ problems is essential for establishing the complete dichotomy for complex-weighted $\hol$ problems. The significance of $\eo$-signatures was first recognized in \cite{cai2013vanishing} during the study of vanishing signatures. Furthermore, \cite{lin2018holantnonnegative} suggests that classifying $\#\eo$ problems may complement the tensor decomposition lemma for complex-weighted signatures. Most directly, \cite{shao2020realholant} demonstrates that research on corresponding $\#\eo$ problems (defined by signatures with ARS property), initially conducted in \cite{cai2020beyond}, constitutes a crucial component of the proof for the real-weighted $\hol$ dichotomy. Additionally, the eight-vertex model dichotomy presented in \cite{caifu2023eightvertex}, where the problem is defined by a single quaternary signature whose support confined to $\eoe\cup\{0000,1111\}$, builds upon the six-vertex model dichotomy in \cite{cai2018complexity}. We believe that some fundamental obstacles to establishing complete complexity classifications for complex-weighted $\hol$ problems lie hidden within complex-weighted $\#\eo$ problems.

\subsection{Our results}\label{section:resulteasy}

All results in this paper target and hold for complex-weighted $\#\eo$ by default. The study begins with analyzing low-arity signatures, a prevalent research focus in counting problems. This approach is justified because an $\eo$ signature $f$ with arity greater than 4 can generate various quaternary and binary signatures through gadget construction.
Our first result is stated below, with its detailed version presented in Theorem \ref{arity4setdichotomy}.

\begin{theorem}\label{thm:4+2}
    Suppose $\mathcal{F}$ is a finite set of $\eo$ signatures with arity less or equal than 4. Then $\#\eo(\mathcal{F})$ is either polynomial time computable or \#P-hard. The classification criterion is explicit.
\end{theorem}

To establish a complete complexity classification for $\#\eo$ problems, we build upon Theorem \ref{thm:4+2} as a foundation. We note that the tractable cases exhibit a more complex structure compared to those in the six-vertex model dichotomy.

Our second result concerns pure signatures. For a set $S$ of 01-strings with length $k$, the \textit{affine span} of $S$ is defined as the minimal affine subspace containing $S$. Given a signature $f$, we denote by $\mathrm{Span}(f)$ the affine span of its support $\su(f)$. Given the affine span of a signature $f$'s support, we call $f$ a pure-up (resp. pure-down) signature if it is contained in $\eog$ (resp. $\eol$). 

Our investigation originates from a quaternary signature $f$ with $\su(f)=\{1100,1010,1001\}$. Its affine span $\{1100,1010,1001,1111\}\subseteq\hw{\geq}$ motivates modifying $f$ to $f'$ by assigning a nonzero value at input 1111. Building on \cite{cai2013vanishing}'s result that $\hol(\neq_2\mid[0,0,a,b,c])\equiv_T\hol(\neq_2\mid[0,0,a,0,0])$ (since $\eosg$ strings do not contribute to the partition function), we establish $\#\eo(f')\equiv_T\#\eo(f)$. This support augmentation technique fundamentally motivates our study of pure signatures. Notably, Theorem~\ref{thm:pureeasy}'s tractable cases generalize those in Theorem~\ref{thm:4+2}.

\begin{theorem}\label{thm:pureeasy}
     Suppose $\mathcal{F}$ is a finite set of pure-up (pure-down) $\eo$ signatures. Then $\#\eo(\mathcal{F})$ is either polynomial time computable or \#P-hard. The classification criterion is explicit.
\end{theorem}

The detailed statement of Theorem~\ref{thm:pureeasy} appear in Theorem \ref{thm:puredichotomy}. To achieve tractability, a condition called the \textit{type condition} must be satisfied in Theorem~\ref{thm:pureeasy}. This condition also plays a crucial role in the subsequent algorithm. Additionally, we define a property called \textit{rebalancing}.
Intuitively, a signature $f$ is 0-rebalancing (or 1-rebalancing) if fixing some of its variables to 0 (respectively 1) naturally forces an equal number of variables to be fixed to 1 (respectively 0). When both the type condition and rebalancing condition are satisfied, a polynomial-time algorithm becomes applicable. We emphasize that this algorithm represents the most significant algorithmic contribution of this paper. Our third result is stated below, with a formal version in Theorem~\ref{thm:rebaalg}.
\begin{theorem}
    $\#\eo(\mathcal{F})$ is polynomial time computable, if every $f\in\mathcal{F}$ is \ba[0] (or \ba[1]), and satisfies the type condition.
    \label{thm:rebaeasy}
\end{theorem}

We note that very recently, and independently of our work, \cite{shao2024eulerian} proposed a polynomial-time algorithm for 01-weighted $\#\eo$ parameterized by $\mathcal{F}$ consisting of $\delta_1$-affine (resp. $\delta_0$-affine) signatures and affine signatures, termed the chain reaction algorithm. Moreover, \cite{shao2024eulerian} establishes a novel connection between the base level of $\delta_1$-affine and $\delta_0$-affine signatures (the $\delta_1$-affine and $\delta_0$-affine kernels) and the Hadamard code, while proving \#P-hardness when $\delta_1$-affine and $\delta_0$-affine signatures are mixed. Although these results do not constitute a complete dichotomy and primarily apply to 0-1 weighted signatures, they capture some of the significant components of our dichotomy and prove highly valuable, as demonstrated in this work.

Furthermore, building upon this article and utilizing our dichotomy for pure signatures \cite{meng2025fpnp}, a complete dichotomy for complex-weighted $\#\eo$ has been established. However, this dichotomy is an $\text{FP}^\text{NP}$ versus \#P dichotomy, indicating that $\#\eo$ problems are either in $\text{FP}^\text{NP}$ or \#P-hard. While this dichotomy provides valuable insights and partially clarifies the complexity classification, it does not yield new polynomial-time algorithms for $\#\eo$ and consequently does not include our algorithm for rebalancing signatures.

\subsection{Our methods}
As presented in Section~\ref{section:resulteasy}, this article contains three parts. In each part, our main objective is to characterize different classes of signature sets. Indeed, the principal difficulty lies in providing descriptions of tractable cases that are sufficiently precise to capture all tractable scenarios. Following this characterization, we separately prove both tractability and \#P-hardness results.

In each part, we perform detailed characterizations of signature sets, including comprehensive classifications and rigorous examinations of their properties. Specifically, we focus on three key aspects: the closure property, the properties of the affine span, and signature reducibility.

We further summarize that each tractable case in this work must satisfy two requirements: the support requirement and the type requirement. The support requirement specifies that each signature's support must possess a particular form. This specific form enables a receiving-sending mechanism that is necessary for implementing the corresponding algorithm. The type requirement states that after applying the receiving-sending mechanism to an instance, all resulting signatures must form a tractable case in $\#\csp$. A detailed explanation appears in Section~\ref{resultm}.

For the hardness results, we prove \#P-hardness based on the dichotomies of both the six-vertex model and $\#\csp$, primarily using the gadget construction method introduced in Section~\ref{section:GC}.

\subsection{Organization}

In Section \ref{section:preli}, we introduce preliminaries needed in our proof. We present the detailed version of our main theorem in Section \ref{resultm}. In Sections \ref{section: arity4set}, \ref{section:allup} and \ref{balancing}, we give the proofs of Theorem \ref{thm:4+2}, \ref{thm:pureeasy} and \ref{thm:rebaeasy} respectively. We conclude our result in Section \ref{sec:ccls}. 

\section{Preliminaries}\label{section:preli}

\subsection{Definitions and notations}

For a $01$-string $\alpha$, we use $\overline\alpha$ to denote the \textit{dual string} of $\alpha$, which is obtained by flipping every bit of $\alpha$. For $s\in \{0,1\}$, we use $\#_s(\alpha)$ to denote the number of $s$'s in $\alpha$. Under this notation we have:

$$\eoe=\{\alpha \mid \#_1(\alpha)=\#_0(\alpha)\}, \quad \eog=\{\alpha \mid \#_1(\alpha)\geq \#_0(\alpha)\}.$$

A signature $f$ is called a \textit{$\eoe$ signature} (precisely an $\eo$ signature) if $\su(f) \subseteq \eoe$. The term ``$\eo$ signature'' was originally introduced in \cite{cai2020beyond} as shorthand for \textit{Eulerian Orientation signature}.

If $\su(f) \subseteq \eog$, we call $f$ an \textit{$\eog$ signature}. Other similar notations are defined analogously by replacing ``$\geq$'' with ``$\leq$'', ``$<$'', or ``$>$'' in both names and defining formulae.

We use $[f_0, f_1, \ldots, f_r]$ to denote a \textit{symmetric signature} $f$ of arity $r$, where $f_i$ gives the evaluation for all input strings with Hamming weight $i$. 

By exchanging the values of the symbols in the Boolean domain, we obtain the \textit{dual signature} $\widetilde f=[\widetilde f_0, \widetilde f_1, \ldots, \widetilde f_r]=[f_r, f_{r-1}, \ldots, f_0]$. A signature $f$ is called \textit{self-dual} if $f=\widetilde f$. The concepts of "dual" and "self-dual" can be naturally generalized to asymmetric signatures, sets of signatures, Boolean domain counting problems, and tractable classes in dichotomy theorems.

We use $[\varphi]$ to denote the truth value of a logical statement $\varphi$. It takes values from $\{\text{True}, \text{False}\}$, or equivalently the isomorphic integer set $\{1,0\}$.
For example, $[x \neq y]$ equals True (or 1) when $x\neq y$, and False (or 0) otherwise. Thus both $[x\neq y]$ and $\neq_2(x,y)$ represent the value of the binary function $[0,1,0]$ evaluated at $(x,y)$.

Several sets of frequently used signatures are defined as follows. $\Delta_0$ represents the unary signature $[1,0]$, with its dual $\Delta_1=[0,1]$. The binary disequality signature $[0,1,0]$ is denoted by $\neq_2$ and serves as an example of self-dual signatures. We use $\mathcal{EQ}$ to denote the set of equality signatures where $\mathcal{EQ}=\{=_1,=_2,\ldots,=_r,\ldots\}$, with $=_r$ denoting an arity-$r$ signature $[1,0,\ldots,0,1]$. In other words, $=_r$ evaluates to 1 when the input is all 0s or all 1s, and 0 otherwise.

A \textit{disequality signature} of arity $2d$, denoted by $\neq_{2d}$, evaluates to 1 when $x_1=x_2=\cdots=x_d\neq x_{d+1}=x_{d+2}=\cdots=x_{2d}$, and 0 otherwise. We define $\mathcal{DEQ}=\{\neq_2,\neq_4,\ldots,\neq_{2n},\ldots\}$ as the set of all disequality signatures. Moreover, $\mathcal{DEQ}$ is closed under variable permutation. That is, for any $\pi\in S_{2d}$, a signature $f$ that evaluates to 1 when $x_{\pi(1)}=x_{\pi(2)}=\cdots=x_{\pi(d)}\neq x_{\pi(d+1)}=x_{\pi(d+2)}=\cdots=x_{\pi(2d)}$ (and 0 otherwise) is also considered a disequality signature.

An alternative definition states that a signature $f$ is a disequality signature if: (1) it is an $\eo$ signature of arity $2d$, (2) there exists $\alpha\in\eoe$ such that $\su(f)\subseteq\{\alpha,\overline\alpha\}$, and (3) $f(\alpha)=f(\overline\alpha)=1$. The equivalence of these definitions can be readily verified.
Similarly, a signature $f$ is called a \textit{generalized disequality signature} (denoted by $\neq_{2d}^{a,b}$) if the signature satisfies (1), (2) and (3') $f(\alpha)=a$, $f(\overline\alpha)=b$.

We use $\otimes$ to denote tensor multiplication. When no ambiguity arises, we sometimes use $f$ to denote the signature set $\{f\}$. We use $\leq_T$ and $\equiv_T$ to respectively denote polynomial-time Turing reductions and equivalences. 
\subsection{Relationships between counting problems}\label{section:relation}

 We say a framework $A$ is more ``general'' than $B$, denoted as $B\preceq A$, if each problem in $B$ can be transformed into a problem in $A$ in polynomial time. The relationship between the counting frameworks are concluded in the following lemma.
\begin{lemma}
    $$\#\textsf{GH}\preceq \#\csp\preceq \#\eo\preceq \hol$$
    \label{lem:relation}
\end{lemma}

Furthermore, $\hol$ constitutes a strictly more expressive framework than $\#\textsf{GH}$, as evidenced by the fact that while counting perfect matchings cannot be represented as a $\#\textsf{GH}$ problem \cite{cai2022perfect}, it can be formulated as a $\hol$ problem.

Lemma \ref{lem:relation} follows directly from Lemmas \ref{lem:ghcsp}--\ref{thm:csp=eom}.
\begin{lemma}
    $\#\textsf{GH}(f)\equiv_T\#\csp(\{f\}).$
    \label{lem:ghcsp}
\end{lemma}

\begin{lemma}[\cite{cai2020beyond}]
    $\#\eo(\mathcal{F})\equiv_T\hol(\neq_2\mid\mathcal{F})$.
    \label{lem:eo=hol}
\end{lemma}

\begin{lemma}[\cite{shao2020realholant}]
    $\hol(\neq_2\mid\mathcal{F})\equiv_T \hol(Z^{-1}\mathcal{F})$
\end{lemma}

\begin{lemma}[{\cite[Theorem 6.1]{cai2020beyond}}]  \label{thm:csp=eom}
    $\#\csp(\mathcal{F})\equiv_T\#\eo(\pi(\mathcal{F})).$
\end{lemma}

We now clarify the notation used in the preceding lemmas, which will also appear throughout this article. Here, $Z^{-1}=\frac{1}{\sqrt{2}}\left(\begin{matrix}
1 & -\mathfrak{i} \\
1 & \mathfrak{i}
\end{matrix}\right)$, and $Z^{-1}\mathcal{F}$ denotes the signature set derived from $\mathcal{F}$ via the holographic transformation defined by $Z^{-1}$ \cite{valiant2008holographic,cai2007valiant}. 

While we omit both the formal definition of holographic transformations and the proof of this lemma since they are not central to our discussion, we strongly recommend readers consult \cite{cai2017complexity}[Section 1.3.2] for deeper understanding, as holographic transformations constitute a powerful technique in counting complexity theory.

Let $f$ be an $\eo$ signature of arity $2d$ with variable set $\text{Var}(f)=\{x_1, x_2, \cdots, x_{2d}\}$. For any pairing $P=\{\{x_{i_1},x_{i_2}\},\{x_{i_3},x_{i_4}\},\cdots,\{x_{i_{2d-1}},x_{i_{2d}}\}\}$ of $\text{Var}(f)$, we define $\eom[P]$ as the subset of $\{0,1\}^{2d}$ corresponding to $P$:

$$\eom[P]=\{\alpha \mid \alpha \in \{0,1\}^{2d}, \alpha_{i_1} \neq \alpha_{i_2}, \cdots, \alpha_{i_{2d-1}} \neq \alpha_{i_{2d}}\},$$
equivalently expressed as $\eom[P]=\su([x_{i_1}\neq x_{i_2}] \cdot [x_{i_3}\neq x_{i_4}] \cdots [x_{i_{2d-1}}\neq x_{i_{2d}}])$.

When $\su(f) \subseteq \eom[P]$, we call $f$ an \textit{$\eom[P]$ signature} relative to $P$. If such a pairing $P$ exists, we call $f$ a \textit{pairwise opposite signature} (identical to Definition 5.5 in \cite{cai2020beyond}), or simply an \textit{$\eom$ signature}.

For an $\eom[P]$ signature $f$ with $\eom[P]=\su([x_{i_1}\neq x_{i_2}] \cdot [x_{i_3}\neq x_{i_4}] \cdots [x_{i_{2d-1}}\neq x_{i_{2d}}])$, we define $\tau_f$ (or $\tau(f)$) as the set of arity-$d$ signatures satisfying:

$$ \tau_f = \{g\mid g(x_{j_1},x_{j_3},...,x_{j_{2d-1}})=f(x_{i_1},1-x_{i_1}, x_{i_3},1-x_{i_3},\ldots,x_{i_{2d-1}},1-x_{i_{2d-1}}),$$
$$x_{j_1}\in\{x_{i_1},x_{i_2}\},\cdots,x_{j_{2d-1}}\in\{x_{i_{2d-1}},x_{i_{2d}}\}\}.$$

Since $P$ allows swapping variables within pairs, we have $|\tau_f|\leq 2^d$.

The inverse mapping $\pi$ is defined for any arity-$d$ signature $g$ as:
$$\pi_g(x_1,y_1,x_2,y_2,\ldots, x_d, y_d)= g(x_1,x_2,\ldots,x_d) \cdot [x_1 \neq y_1] \cdots [x_d \neq y_d],$$
where $\pi_g$ (or $\pi(g)$) is clearly an $\eom[P]$ signature with $P=\{\{x_1,y_1\},\{x_2,y_2\},\cdots,\{x_d,y_d\}\}$. Extending this to signature sets, we define $\pi(\mathcal{F})=\{\pi(f)\mid f\in\mathcal{F}\}$. 

The set $\tau(\pi(g))$ contains $g$ and all signatures obtainable from $g$ by variable flips using binary disequality. This establishes an equivalence between $\#\eo$ and $\#\csp$ with free binary disequality. As established in \cite{cai2020beyond}, Lemma \ref{thm:csp=eom} comes from this observation and the complexity equivalence between $\#\csp$ and $\#\csp$ with free binary disequality.

\subsection{Gadget construction and signature matrix}\label{section:GC}

This subsection introduces two key concepts: gadget construction and signature matrices. Gadget construction serves as a fundamental reduction technique in complexity classification, while signature matrices offer an intuitive representation of signatures and reveal connections between gadget construction and matrix multiplication.

Given a signature set $\mathcal{F}$, an $\mathcal{F}$-gate resembles a signature grid $\Omega(G,\pi)$ from Definition~\ref{def:holant}, but with $G=(V,E,D)$ where $V$ represents vertices, $E$ denotes internal edges (each connecting two vertices) and $D$ denotes dangling edges (each connecting one vertex).
An $\mathcal{F}$-gate essentially defines a signature whose variables correspond to edges in $D$. For $|E|=n$ and $|D|=m$, with edges in $E$ representing variables $\{x_1,...,x_n\}$ and edges in $D$ representing $\{y_1,...,y_m\}$, the $\mathcal{F}$-gate's signature $f$ satisfies:

$$f(y_1,...,y_m)=\sum_{\sigma:E\rightarrow\{0,1\}}\prod_{v\in V}f_v(\hat{\sigma}|_{E(v)\cup D(v)})$$

where $(y_1,...,y_n)\in\{0,1\}^n$ assigns values to dangling edge variables, $\hat{\sigma}$ extends 
$\sigma$ with these assignments and $f_v$ denotes the signature at vertex $v$.

We say $f$ is realizable from $\mathcal{F}$ via gadget construction when $f$ is an $\mathcal{F}$-gate's signature. 
Crucially, when $f$ is realizable from $\mathcal{F}$, \cite[Lemma 1.3]{cai2017complexity} establishes:

$$\hol(\mathcal{F})\equiv_T\hol(\mathcal{F}\cup\{f\})$$

For $\#\eo$ problems, $\mathcal{F}$-gates have a modified definition. Since Lemma~\ref{lem:eo=hol} shows $\#\eo(\mathcal{F})\equiv_T\hol(\neq_2\mid\mathcal{F})$, we add $\neq_2$ vertices to internal edges
and treat the gadget as a standard $\hol$ gate.
Throughout this paper, $\mathcal{F}$-gates in $\#\eo$ contexts follow this definition unless stated otherwise.

Next, we introduce the matrix form of signatures, which provides an orderly way to enumerate all possible values of a signature. 
Let $f$ be an arbitrary signature mapping from $\{0,1\}^r$ to $\mathbb{C}$. First we fix an ordering of the variables as $x_1,x_2,\ldots,x_r$. The matrix form of $f$ (or signature matrix of $f$ \cite{cai2017complexity}) with parameter $l$ is a $2^l\times 2^{r-l}$ matrix for some integer $0\leq l\leq r$, where the values of $x_1,x_2,\ldots,x_l$ determine the row index and the values of $x_{l+1},x_{l+2},\ldots,x_r$ determine the column index. 
This matrix is denoted by $M(f)_{x_1x_2\cdots x_l,x_{l+1}x_{l+2}\cdots x_r}$. For convenience, we sometimes omit $(f)$ and write simply $M_{x_1\cdots x_l,x_{l+1}\cdots x_r}$, or use the abbreviated form $M_f$ when no ambiguity arises. For signatures with even arity, we typically choose $l=\frac{r}{2}$.

\begin{example}[Signature matrix]
    If $f$ is a binary signature and $f=(f_{00},f_{01},f_{10},f_{11})$, then $M_f=M_{x_1,x_2}=\left(\begin{matrix}
f_{00} & f_{01} \\
f_{10} & f_{11}
\end{matrix}\right)$.

If $g$ is a quaternary signature, then

$$M_g=M_{x_1x_2,x_3x_4}=\left(\begin{matrix}
g_{0000} & g_{0001} & g_{0010} & g_{0011}\\
g_{0100} & g_{0101} & g_{0110} & g_{0111}\\
g_{1000} & g_{1001} & g_{1010} & g_{1011}\\
g_{1100} & g_{1101} & g_{1110} & g_{1111}
\end{matrix}\right).$$

\end{example}

We now establish the connection between matrix multiplication and gadget construction. Consider two $\mathcal{F}$-gates $f$ and $g$ with arities $n$ and $m$ respectively, where $\text{Var}(f)=\{x_1,x_2,\ldots,x_n\}$ and $\text{Var}(g)=\{y_1,y_2,\ldots,y_m\}$. By connecting a subset of their dangling edges - specifically pairing $\{x_{n-l+1},\ldots,x_n\}$ with $\{y_1,y_2,\ldots,y_l\}$ for some positive integer $l$, we construct a new $\mathcal{F}$-gate $h$.

The resulting signature $h$ has variables $\text{Var}(h)=\{x_1,\ldots,x_{n-l},y_{l+1},\ldots,y_m\}$. Through direct comparison of gadget computation and matrix multiplication, we obtain:

$$M_h=M_{x_1\cdots x_{n-l},y_{l+1}\cdots y_m}=M_{x_1\cdots x_{n-l},x_{n-l+1}\cdots x_n}\cdot M_{y_1\cdots y_l,y_{l+1}\cdots y_m}=M_f\cdot M_g.$$

For $\#\eo$ problems, the only modification is that edges now represent $\neq_2$ instead of $=_2$, yielding:

$$M_h=M_f\cdot \neq_2^{\otimes l}\cdot M_g.$$

This relationship will be frequently employed in subsequent discussions without further explanation.

We highlight two special forms of gadget construction: self-loop and pinning. In $\#\eo$ problems, adding a self-loop to signature $f$ involves selecting two variables $x_1$ and $x_2$ and connecting them with an internal edge, yielding new signature $f^{x_1\neq x_2}$. For $\text{Var}(f)=\{x_1,x_2,\ldots,x_n\}$, since the internal edge represents $\neq_2$, we obtain:

$$f^{x_1\neq x_2}(x_3,\ldots,x_n)=f(0,1,x_3,\ldots,x_n)+f(1,0,x_3,\ldots,x_n).$$

This operation can be generalized using $\bi{a}{b}$. Instead of directly connecting $x_1$ and $x_2$, we introduce a new vertex assigned $\bi{a}{b}$ and connect $x_1$ and $x_2$ to its dangling edges.
The two possible connection orders yield either:
$f'(x_3,\ldots,x_n)=af(0,1,x_3,\ldots,x_n)+bf(1,0,x_3,\ldots,x_n)$
or
$f'(x_3,\ldots,x_n)=bf(0,1,x_3,\ldots,x_n)+af(1,0,x_3,\ldots,x_n)$.

Pinning represents the second key operation. While standard $\hol$ and $\#\csp$ problems use $\Delta_0=(1,0)$ or $\Delta_1=(0,1)$ to fix variables, $\#\eo$ problems employ $\bi{1}{0}$ self-loops. For $x,y\in\text{Var}(f)$, pinning simultaneously fixes $x=1$ and $y=0$, preserving the $\eo$ property. The resulting signature is denoted $f^{x=1,y=0}$. Throughout this paper, "pinning signature" refers unambiguously to $\bi{1}{0}$.

\subsection{Two known dichotomies for \#CSP and Six-vertex model} \label{sixvertexchapter}

In this section, we introduce the dichotomies for $\#\csp$ and six-vertex model. We start from defining tractable classes $\mathscr{A}$ and $\mathscr{P}$. 
Let $X$ be a $(d+1)$-dimensional column vector $X=(x_1,x_2,\cdots,x_d,1)^T$ over $\mathbb{F}_2$. Suppose $A$ is a matrix over $\mathbb{F}_2$. The function $\chi_{AX}$ takes value 1 when $AX=\mathbf{0}$, otherwise 0. In other words, $\chi_{AX}$ describe an affine relation on variables $x_1,x_2,\cdots,x_d$.

\begin{definition}[\cite{cai2014complexity}]
    We denote by $\mathscr{A}$ the set of signatures which have the form $\lambda\cdot\chi_{AX}\cdot\mathfrak{i}^{L_1(X)+L_2(X)+\cdots+L_n(X)}$, where $\mathfrak{i}=\sqrt{-1}$, $\lambda\in\mathbb{C}$, $n\in\mathbb{Z}_+$, each $L_j$ is a 0-1 indicator function of the form $\langle \alpha_j,X \rangle$, where $\alpha_j$ is a $(d+1)$-dimensional vector over $\mathbb{F}_2$, and the dot product $\langle \cdot,\cdot \rangle$ is computed over $\mathbb{F}_2$.
\end{definition}

\begin{definition}[\cite{cai2014complexity}]
    We denote by $\mathscr{P}$ the set of all signatures which can be expressed a product of unary signatures, binary equality signatures ($=_2$) and binary disequality signatures ($\neq_2$) (on not necessarily disjoint subsets of variables). 
\end{definition}

As in \cite[chapter 3]{cai2017complexity}, $\mathscr{A}$ and $\mathscr{P}$ are closed under gadget construction. That is, if $\mathcal{F}\subseteq\mathscr{A}$ or $\mathcal{F}\subseteq\mathscr{P}$, then all $\mathcal{F}$-gates are respectively in $\mathscr{A}$ or $\mathscr{P}$. Now we state the dichotomy for $\#\csp$ from \cite[Theorem 3.1]{cai2014complexity}.

\begin{theorem}[\cite{cai2014complexity}]\label{thm:CSPdichotomy}
Suppose $\mathcal{F}$ is a finite set of signatures mapping Boolean inputs to complex numbers. If $\mathcal{F}\subseteq\mathscr{A}$ or $\mathcal{F}\subseteq\mathscr{P}$, then $\#\csp(\mathcal{F})$ is computable in polynomial time. Otherwise, $\#\csp(\mathcal{F})$ is \#P-hard.
\end{theorem}

For \eo\ signatures belonging to $\mathscr{A}$ or $\mathcal{P}$, they satisfy the following properties. 
\begin{lemma}[{\cite[Lemma 5.7]{cai2020beyond}}]  \label{lem: eo+affine= eom}
If an $\eo$ signature $f$ has affine support, then $f$ is a pairwise opposite signature ($\eom$ signature). 
\end{lemma}
\begin{lemma}[\cite{cai2020beyond}]
    $f\in\mathscr{A}$ (resp. $\mathscr{P}$) if and only if $\tau(f)\subseteq  \mathscr{A}$ (resp. $\mathscr{P}$).
\end{lemma}

Now we introduce some definitions for the dichotomy for six-vertex model. $\mathscr{M}$ is the set of all ternary signatures whose supports only contain strings of Hamming weight exactly 1. In other words, a quaternary \eo\ signature $f\in \mathscr{M} \otimes \Delta_1 $, if and only if $\su(f) \subseteq \{1\} \times \{100,010,001\}=\{1100,1010,1001\}$. $\widetilde{\mathscr{M}}$ is defined similarly, except that the Hamming weight of each support string is exactly 2. The dichotomy for six-vertex model is as follows. 

\begin{theorem}[\cite{cai2018complexity}] \label{thm:sixvertexdichotomyMform}
    Let f be a quaternary $\eo$ signature, then $\hol(\neq_2\mid f)$ is \#P-hard except for the following cases:

a. $f\in\mathscr{P}$ (equivalently, $f \in \eom$ and $\tau(f) \subseteq \mathscr{P}$);

b. $f\in\mathscr{A}$ (equivalently, $f \in \eom$ and $\tau(f) \subseteq \mathscr{A}$);

c. $f\in \mathscr{M}\otimes  \Delta_1 $;

d. $f\in  \widetilde{ \mathscr{M}} \otimes  \Delta_0$;

in which cases $\hol(\neq_2\mid f)$ is computable in polynomial time.
\end{theorem}

\subsection{Insights from tractable cases}

Starting from this section, this article presents a number of insights pertaining to the tractable cases within each dichotomy. The objective is to generalize these tractable cases at a high level of abstraction, although some statements may be informal. It is our hope that these generalized characteristics, derived from the tractable cases, will prove a valuable point of reference for future research.

In the dichotomy of $\#\csp$, both  $\mathscr{A}$ and $\mathscr{P}$  are self-dual signature sets.  
All signatures in $\mathscr{A}$ and $\mathscr{P}$ have affine supports, so having affine support is a necessary condition for tractability here, and we consider this condition as a support requirement. 
Informally speaking, there are two ways to reach tractability based on the support requirement, namely $\mathscr{A}$ type and $\mathscr{P}$ type requirements.  
The two steps explanation for the tractability part in Theorem \ref{thm:CSPdichotomy}, will be used as an important template in explaining the following results. 
We hope to exhibit an informal evolution process from the known tractable cases to the new tractable cases in this way.

In the dichotomy for six-vertex model, the support requirement is an union of two parts. The first part is that the support must be affine. Regardless of variable permutations, it is exactly an affine subspace of $Q=\{0101,1001,0110,1010\}$. This corresponds to the tractable cases a and b in Theorem \ref{thm:sixvertexdichotomyMform}.

The second part is that the support must be an subset of $Q_1=\{1100,1010,1001\}$ or its dual, $Q_0=\{0011,0101,0110\}$.
This requirement itself is enough for tractability.
Noticing that if each signature $f$ in the instance satisfies that $\su(f)=Q_1$, then one of its variables is fixed to 1, say $x_1$. 
Transmitted by $\neq_2$, as well as the general binary disequality signatures in the instance, another edge $x_2$ of a signature $g$ will be fixed to 0. 
 Then the valid part of $\su(g)$, the set of strings that can appear in an assignment with non-zero evaluation, is of size at most two and meets the affine support requirement. We denote this process as the \textit{receiving-sending mechanism}. This mechanism directly corresponds to the chain reaction algorithm in \cite{shao2024eulerian}.
This process may repeat many rounds, until each signature collapses to $\eom$. 
At each round, an opposite pair of variables is somehow fixed by this mechanism. And the process continues until the rest signatures in the whole instance are in $\eom$. Hence, this more general support requirement, can force a signature in an instance to work exactly like some $\eom$ signature. 
The analysis of the dual case $Q_0$ is similar.

\section{New notations and our results}\label{resultm}
 
 There are three parts of results in this paper, which are introduced in Subsection \ref{2+4}, \ref{purewater}, \ref{rererebl} respectively. In Subsection \ref{4056}, we introduce two strange signatures, which serve as a motivation for some of our study. Finally in Subsection \ref{relations} we introduce the relations between these results.

\subsection{A set of binary or quaternary signatures}\label{2+4}

The first part of our results generalizes six-vertex model to $\#\eo$ problems defined by a set of $\eo$ signatures of arity no more than 4.  We define $$\mathscr{M}_{\mathscr{A}}=\{f\in\mathscr{M}\mid\text{the quotient of any two non-zero values of } f \text{ belongs to } \{\pm1,\pm\mathfrak{i}\}\}.$$ $\mathscr{M}_{\mathscr{A}} \otimes \Delta_1$ can be seen as $\mathscr{M} \otimes \Delta_1$ with the $\mathscr{A}$ type requirement. 
Similarly, we define $$\widetilde{\mathscr{M}_{\mathscr{A}}}=\{f\in\widetilde{\mathscr{M}}\mid\text{the quotient of } \text{any two } \text{non-zero values } \text{of }f \text{ belongs to } \{\pm1,\pm\mathfrak{i}\}\}.$$
The detailed version of Theorem \ref{thm:4+2} is as follows.

\begin{theorem}\label{arity4setdichotomy}
Suppose $\mathcal{F}$ is a set of \eo\ signatures with arity less than or equal to 4. Then $\hol(\neq_2\mid\mathcal{F})$ is \#P-hard unless one of the following conditions holds:

a. $\mathcal{F}\subseteq \mathscr{P}\cup(\mathscr{M}\otimes\Delta_1)$;

b. $\mathcal{F}\subseteq \mathscr{A}\cup(\mathscr{M}_{\mathscr{A}}\otimes\Delta_1)$;

c. $\mathcal{F}\subseteq \mathscr{P}\cup(\widetilde{\mathscr{M}}\otimes\Delta_0)$;

d. $\mathcal{F}\subseteq \mathscr{A}\cup(\widetilde{\mathscr{M}_{\mathscr{A}}}\otimes\Delta_0)$,

in which cases the problem can be computed in polynomial time.
\end{theorem}

We summarize current notations as the following tree.  A specific example of the $\mathscr{A}$ type requirement is also presented. When a set of signatures can meet all requirements on a path from the root to the four nodes at the bottom, the corresponding $\#\eo$ problem is tractable; otherwise the problem is \#P-hard.

\begin{center}
\begin{tikzpicture}[
  minimum size=5mm,
  sibling distance=5cm,
  ]
\node[draw] {Quaternary or binary \eo\ signatures}
child {node[draw] {Meets $Q$ or $Q_1$ support req.}
child{node[draw]{Meets $\mathscr{P}$ type req.}}
child{node[draw][below=0.6]{Meets $\mathscr{A}$ type req. ($ \mathscr{A}\cup(\Delta_1 \otimes \mathscr{M}_{\mathscr{A}})$)} }
}
child {node[draw](1){Meets $Q$ or $Q_0$ support req.}
child{node[draw,xshift= 1 cm]{Meets $\mathscr{P}$ type req.}}
child{node[draw]{Meets $\mathscr{A}$ type req.}}
};
\end{tikzpicture}
\end{center}

\subsubsection{Insights from tractable cases}
This dichotomy exhibits two dual support requirements represented by $\{Q,Q_1\}$\footnote{Each signature's support must be a subset of either $Q$ or $Q_1$.} and $\{Q,Q_0\}$. Tractability requires satisfying either support requirement along with either an $\mathscr{A}$-type or $\mathscr{P}$-type requirement.

A key distinction between Theorems~\ref{thm:sixvertexdichotomyMform} and~\ref{arity4setdichotomy} lies in the emergence of $\mathscr{M}_{\mathscr{A}}$ in the latter. It can be verified that:
\begin{itemize}
    \item For $Q$ signatures, $\mathscr{A}$-type and $\mathscr{P}$-type requirements produce the first two tractable classes in Theorem~\ref{thm:sixvertexdichotomyMform};
    \item $Q_1$ signatures under $\mathscr{A}$-type requirements automatically satisfy $\mathscr{P}$-type conditions, explaining the third tractable case;
    \item The dual $Q_0$ case similarly generates the fourth tractable case.
\end{itemize}

Theorem~\ref{arity4setdichotomy} demonstrates the crucial divergence between $\mathscr{A}$-type and $\mathscr{P}$-type requirements when handling mixed signature sets.

\subsubsection{Proof outline of hardness result}

We say two signature sets $A,B$ mix (in $\mathcal{F}$) if there exists $f,g\in \mathcal{F}$ such that $f\in A-B$ and $g\in B-A$.
We first prove 3 no-mixing lemmas, mainly through signature classification on support size and gadget construction, showing the following cases are \#P-hard:
\begin{enumerate}
    \item $\mathscr{M}\otimes\Delta_1$ and $\widetilde{\mathscr{M}}\otimes\Delta_0$ mix;
    \item $\mathscr{A}$ and $\mathscr{P}$ mix;
    \item $\mathscr{A}-\mathscr{P}$ and $(\mathscr{M}\otimes\Delta_1)-(\mathscr{M_A}\otimes\Delta_1)$ mix;
    \item $\mathscr{A}-\mathscr{P}$ and $(\widetilde{\mathscr{M}}\otimes\Delta_0)-(\widetilde{\mathscr{M_A}}\otimes\Delta_0)$ mix.
\end{enumerate}

Suppose $\mathcal{F}$ is a set of binary and quaternary $\eo$ signatures. Suppose the problem is not \#P-hard. By Theorem \ref{thm:sixvertexdichotomyMform}, we can assume that $\mathcal{F}\subseteq\mathscr{A}\cup\mathscr{P}\cup(\mathscr{M}\otimes\Delta_1)\cup(\widetilde{\mathscr{M}}\otimes\Delta_0)$.

 As $\mathscr{M}\otimes\Delta_1$ and $\widetilde{\mathscr{M}}\otimes\Delta_0$ do not mix, either $\mathcal{F} \subseteq \mathscr{A}\cup\mathscr{P}\cup(\mathscr{M}\otimes\Delta_1)$ or $\mathcal{F} \subseteq \mathscr{A}\cup\mathscr{P}\cup(\widetilde{\mathscr{M}}\otimes\Delta_0)$. Due to the dual property, it is sufficient to focus on the former case.

 As $\mathscr{A}$ and $\mathscr{P}$ do not mix, either $\mathcal{F} \subseteq \mathscr{P}\cup(\mathscr{M}\otimes\Delta_1)$ or $\mathcal{F} \subseteq \mathscr{A}\cup(\mathscr{M}\otimes\Delta_1)$. The former case is exactly tractable case a in Theorem~\ref{arity4setdichotomy}. We focus on the latter case.

 As $\mathscr{A}-\mathscr{P}$ and $(\mathscr{M}\otimes\Delta_1)-(\mathscr{M_A}\otimes\Delta_1)$ do not mix, either $\mathcal{F} \subseteq \mathscr{A}\cup(\mathscr{M_A}\otimes\Delta_1)$ or $\mathcal{F} \subseteq (\mathscr{A}\cap \mathscr{P})\cup(\mathscr{M}\otimes\Delta_1)$. The former case is exactly tractable case b in Theorem~\ref{arity4setdichotomy}, and the latter case is subsumed by  tractable case a.

\subsection{Pure signatures}\label{purewater}

In this section we focus on pure signatures, which can be seen as a generalization of the quaternary signatures in $\mathscr{M}\otimes\Delta_1$ and $\widetilde{\mathscr{M}}\otimes\Delta_0$.

\begin{definition}\label{def:pure}
    An $\eo$ signature $f$ is pure-up, if $\text{Span}(f)\subseteq\eog$. A signature set $\mathcal{F}$ is  pure-up, if each signature in it is pure-up. Similarly, An $\eo$ signature $f$ is pure-down, if $\text{Span}(f)\subseteq\eol$. A signature set $\mathcal{F}$ is pure-down, if each signature in it is pure-down.

    These two categories are collectively termed \textit{pure signatures}.
\end{definition}
We summarize current notations as the following tree.

\begin{center}
\begin{tikzpicture}[
  minimum size=5mm,
  sibling distance=4cm,
  ]
\node[draw](1){Even arity signatures}
child {node[draw] (5)[xshift=-2cm]{$\eol$}}
child {node[draw](2) {$\eo$}
child {node[draw](7) {pure-down}  child{node[draw](3)[below of =2,yshift=-1cm] {$\eom$}}}
child {node[draw](4){pure-up} }
}
child {node[draw](8)[xshift=2cm] {$\eog$}}
;
\draw (3) edge  (4);
\draw[dashed, ->,thick] (7) edge node[sloped,above]{span contained in}   (5);
\draw[dashed, ->,thick] (4) edge  node[sloped,above]{span contained in}  (8);
\end{tikzpicture}
\end{center}

 In the following, we present the dichotomy of pure signatures and the corresponding definitions. It is worth noticing that in this section, pure-up signatures and pure-down signatures never mix in $\mathcal{F}$.

\begin{definition}
  Suppose $f$ is an arity $2d$ $\eo$ signature and $S\subseteq \eoe$. $f|_{S}$ is the restriction of $f$ to $S$, which means when $\alpha\in S$, $f|_{S}(\alpha)=f(\alpha)$, otherwise $f|_{S}(\alpha)=0$.
  
  If for any pairing $P$ of Var$(f)$, $f|_{\eom[P]}\in\mathscr{A}$, then we say that $f$ is $\eom[\mathscr{A}]$.
  
  Similarly, if for any pairing $P$ of Var$(f)$, $f|_{\eom[P]} \in \mathscr{P}$, then $f$ is $\eom[\mathscr{P}]$.
  \label{def:eoaeop}
\end{definition}
\begin{theorem}[The dichotomy for pure-up $\eo$ signatures]
    Suppose $\mathcal{F}$ is a set of pure-up $\eo$ signatures. Then $\#\eo(\mathcal{F})$ is \#P-hard unless all signatures in $\mathcal{F}$ are $\eom[\mathscr{A}]$ or all of them are $\eom[\mathscr{P}]$, in which cases the problem can be computed in polynomial time.
    \label{thm:puredichotomy}
\end{theorem}

\subsubsection{Insights from tractable cases}
We primarily analyze pure-up signatures, noting that pure-down signatures admit dual analysis through symmetric arguments.
Key properties of pure-up signatures reveal that each either belongs to $\eom$, or has at least one variable fixed to 1.
These properties enable application of the active receiving-sending mechanism. Pure signatures thus induce a generalized support requirement with two dual cases (intersecting at the self-dual $\eom$ requirement), where tractability emerges through additional $\eom[\mathscr{A}]$ type or $\eom[\mathscr{P}]$ type requirements.

\subsubsection{Proof outline of hardness result}

For a pure signature $f$of arity $2d$, it can be proved that for an arbitrary pairing $P$ of $\text{Var}(f)$, $f|_{\eom[P]}=f'\otimes\Delta_0^{2d-2k}\otimes\Delta_1^{2d-2k}$ where $f'$ is an $\eom$ signature of arity $2k$. Such $f'$ can be realized by adding $2d-2k$ self-loops, and $f'\in \mathscr{A}$ (or $\mathscr{P}$) if and only if $f|_{\eom[P]}\in \mathscr{A}$ (or $\mathscr{P}$).

    Consequently, if $\mathcal{F}$ is not a subset of $\eom[\mathscr{A}]$ or $\eom[\mathscr{P}]$, then there exist $\eom$ signatures $f'\notin \mathscr{A}$ and $g'\notin \mathscr{P}$ can be realized.  By Lemma \ref{thm:csp=eom}, $\#\eo(\{f',g'\})$ is equivalent to a $\#\csp(\{f'',g''\})$ problem, where $f''\notin \mathscr{A}$ and $g''\notin \mathscr{P}$. By Theorem \ref{thm:CSPdichotomy} $\#\csp(\{f'',g''\})$ is \#P-hard, and consequently $\#\eo(\mathcal{F})$ is \#P-hard.
\subsection{Rebalancing signatures}\label{rererebl}
In this section, we define a property  named $rebalancing$ for $\eo$ signatures.

\begin{definition}[Rebalancing]
An $\eo$ signature $f$ of arity $2d$, is called \ba[0](\ba[1] respectively), when the following  recursive conditions are met.
\begin{itemize}
    \item $d=0$: No restriction.
    \item $d\ge 1$: For any variable $x$ in $X=\text{Var}(f)$, there exists a variable $y=\psi(x)$ different from $x$, such that for any $\alpha\in\{0,1\}^X$, if $\alpha_x=\alpha_y=0$($\alpha_x=\alpha_y=1$ respectively) then $f(\alpha)=0$. Besides, the arity $2d-2$ signature $f^{x=0,y=1}$ is 0-rebalancing($f^{x=1,y=0}$ is \ba[1] respectively).
\end{itemize}
For completeness we view all nontrivial signatures of arity 0, which is a non-zero constant, as \ba[0](\ba[1]) signatures. For the $d\ge 1$ case, $\psi$ is said to be the first level mapping of $f$.
\label{def:reba}
\end{definition}

Tractability can also be obtained based on \ba[0]/\ba[1] property.

\begin{theorem}
If $\mathcal{F}$ is composed of \ba[0](\ba[1]) $\eom[\mathscr{A}]$ signatures, or $\mathcal{F}$ is composed of \ba[0](\ba[1]) $\eom[\mathscr{P}]$ signatures, then $\#\eo(\mathcal{F})$ is polynomial-time computable.
\label{thm:rebaalg}
\end{theorem}

\begin{remark}
It can be verified that every $\eom$ signature satisfies both \ba[0] and \ba[1] conditions. All pure-up signatures are \ba[0] due to the following arguments. When a pure-up signature $f$ satisfies $\su(f)\cap\hw{>}\neq\emptyset$, it necessarily contains a $\Delta_1(x)$ factor. This factor induces a first-level mapping $\psi$ where all variables except $x$ map to $x$, while $\psi(x)$ can be any variable other than itself. This mapping process can be recursively applied until the signature $f$ is transformed into an $\eom$ signature, which remains \ba[0]. Through dual reasoning, all pure-down signatures satisfy \ba[1].
\end{remark}
\subsubsection{Insights from the algorithm}

This section extends the receiving-sending mechanism to rebalancing signatures. While the active receiving-sending mechanism for pure signatures requires initial $\Delta_1$ ($\Delta_0$) triggers, rebalancing signatures may lack such initial conditions. Nevertheless, by assuming fixed values for certain variables, the mechanism still functions through what we term the passive receiving-sending mechanism. Building on this modified mechanism, we develop a polynomial-time algorithm for $\#\eo$ problems defined by \ba[0] or \ba[1] signatures meeting $\mathscr{A}$-type or $\mathscr{P}$-type requirements.

Definition~\ref{def:reba} for $d\ge 1$ specifies two necessary conditions: the existence of mapping $\psi$ (first-level condition) and the recursive condition. Crucially, the first-level condition remains unconditional - the relationship between $x$ and $\psi(x)$ persists even when other variables are fixed. This property enables identification of loops (rather than $\Delta_0$/$\Delta_1$-initiated paths) in the passive mechanism, while the recursive condition ensures the mechanism's iterative applicability.

\subsection{Two strange signatures}\label{4056}

In this subsection we present the definition of $f_{40}$ and $f_{56}$, which are considered as two important examples in our study. Some of our results are motivated by these signatures. We first introduce some notations to describe signatures with large arity and a sparse support.

For matrices $A_{s\times p}$ and $B_{s\times q}$, we use $C_{s\times (p+q)}=[A\ B]$ to denote the matrix $C$  satisfying 
\begin{equation}
    C(i,j)=\begin{cases}
        A(i,j),& 1\le j\le p;\\
        B(i,j-p),& p+1\le j\le p+q.
    \end{cases}
\end{equation}
Informally speaking, $C$ is a concatenation of $A,B$.
We use $A^{\to k}$ to denote the matrix $[A\  A\ ...\ A]$, which is a concatenation of $k$ copies of matrix $A$. We also define the following matrices.

$$H_0=\begin{pmatrix}
    0\\
     0\\
      0\\
       0\\
        0\\
\end{pmatrix}, 
H_2=\begin{pmatrix}
    1&1&1&1&0&0&0&0&0&0\\
    1&0&0&0&1&1&1&0&0&0\\
    0&1&0&0&1&0&0&1&1&0\\
    0&0&1&0&0&1&0&1&0&1\\
    0&0&0&1&0&0&1&0&1&1\\
\end{pmatrix},
H_4=\begin{pmatrix}
    0&1&1&1&1\\
    1&0&1&1&1\\
    1&1&0&1&1\\
    1&1&1&0&1\\
    1&1&1&1&0\\
\end{pmatrix}$$

For a signature $f$ of arity $k$, we use the support matrix $R_{s\times k}$ to describe its support. A string $\alpha\in \{0,1\}^k$ is a support string of $f$ if and only if $\alpha$ equals a row in $R_{s\times k}$. $f_{40}$ has the support matrix $R_{5\times40}=[H_2^{\to 3}\  H_4^{\to 2}]$, 
$f_{56}$ has the support matrix $R_{5\times56}=[H_0\ H_2^{\to 4}\ H_4^{\to 3}]$ and both of them are 0-1 weighted. Though the supports of them are sparse, they are not able to be captured by a number of existing algorithms and hardness results.

\subsection{Relations between results}\label{relations}

The first two parts collectively establish complete dichotomies for $\#\eo$ problems, incorporating both tractability and hardness results: the first for binary and quaternary signature sets, the second for pure-up/pure-down signature sets. The hardness results from the first part form the foundational basis for the hardness classification in complex-weighted $\#\eo$, while those from the second part are equally pivotal for establishing the full $\#\eo$ dichotomy.

The third part focuses on defining the rebalancing property and presenting a polynomial-time algorithm for 0-rebalancing/1-rebalancing signatures with type requirements. Notably, this algorithm subsumes both previous algorithms through its innovative use of support requirements to emulate the $\eom$ property's effects. The case of $f_{40}$ (defined in Section~\ref{4056}) exemplifies this generalization - although not a pure signature, its satisfaction of \ba[0] demonstrates our algorithm's strictly broader applicability compared to the pure signatures' algorithm.

We observe that although the results in \cite{shao2024eulerian} are restricted to 0-1 weighted $\#\eo$ problems, they nevertheless capture essential support requirements for signatures. Specifically, the pure property inherently implies the $\delta$-affine property, and their chain reaction algorithm implements what we term the active receiving-sending mechanism.

The complex-weighted $\#\eo$ dichotomy in \cite{meng2025fpnp} adopts the notation $\forall k \uparrow$ for pure-up signatures and $\forall k \downarrow$ for pure-down signatures, where our Theorem~\ref{thm:puredichotomy} directly resolves "Case 3a" and "Case 4a". Notably, the signatures $f_{40}$ and $f_{56}$ introduced in our work, whose non-trivial nature is evident from our analysis, provide concrete examples for "Case 3c".
\section{Generalization of the six-vertex model to a set of signatures} \label{section: arity4set}

In this section, we prove Theorem \ref{arity4setdichotomy}. Firstly, we introduce some special notations used in this section.

We use $[0,(a,b,c),0,0]$ to denote a ternary signature $f$ satisfying $\su(f)=\{001,010,100\}$ and 
$$f(001)=a,f(010)=b,f(100)=c.$$ 
Similarly, $[0,0,(a,b,c),0]$ denotes a signature $f$ satisfying $\su(f)=\{011,101,110\}$ and 
$$f(011)=a,f(101)=b,f(110)=c.$$

We denote all signatures with the same size of support by adding a subscript to the original signature set. For example, $\mathscr{A}_1=\{f\in\mathscr{A}\mid|\su(f)|=1\}$. We can notice that since all signatures in $\mathscr{P}$ and $\mathscr{A}$ have affine supports, $\mathscr{A}_3$ and $\mathscr{P}_3$ are both empty.

In Section \ref{Characterizations}, we give characterizations of signatures in the tractable cases of Theorem \ref{thm:sixvertexdichotomyMform}. In Sections \ref{subsection:4aritysetalgorithm} and \ref{subsection:4aritysethardness}, we prove the tractability and hardness result in Theorem \ref{arity4setdichotomy} respectively.

\subsection{Characterizations of specific $\eo$ signatures}\label{Characterizations}

In this section, we build a detailed description of quaternary $\eo$ signatures in $\mathscr{P}$, $\mathscr{A}$, $\mathscr{M}\otimes\Delta_1$ and $\widetilde{\mathscr{M}}\otimes\Delta_0$.
All signatures in $\mathscr{A}$ and $\mathscr{P}$ have affine support. By Lemma \ref{lem: eo+affine= eom}, all $\eo$ signatures with affine supports are in $\eom$. For a signature $f$, we use $\text{Span}(f)$ to denote the affine span of $\su(f)$. We state a slightly stronger lemma in the following, which can be proved by the same approach in \cite[Lemma 5.7]{cai2020beyond} as well.

\begin{lemma}\cite{cai2020beyond} \label{lemmaaffineeosupport}
Suppose $f$ is an $\eo$ signature with arity $2d$. If $\text{Span}(f)$ is a linear affine subspace of $\{0,1\}^{2d}$, then $f$ is in $\eom$.

\end{lemma}

Now we characterize the quaternary signatures in $\mathscr{P}$, $\mathscr{A}$, $\mathscr{M}\otimes\Delta_1$ and $\widetilde{\mathscr{M}}\otimes\Delta_0$.

\begin{lemma}\label{lem:arity4inP}
Suppose $f\in \mathscr{P}$ is a non-trivial $\eo$ signature of arity 4. Then one of the following statements holds up to variable reordering:

i. $f\in\mathscr{P}_1$, and $f=\lambda\Delta_0^{\otimes2}\otimes\Delta_1^{\otimes2},\lambda\neq0$. 

ii. $f\in\mathscr{P}_2$, and either (1) or (2) holds.

(1) $f=\Delta_0\otimes\Delta_1\otimes\neq_2^{a,b},ab\neq 0$. 

(2) $f=(\neq_4^{a,b}), ab\neq0.$

iii. $f\in\mathscr{P}_4$, and $f=(\neq_2^{a,b}\otimes\neq_2^{c,d}), abcd\neq0$.
\end{lemma}
\begin{proof}
When $|\su(f)|=1$, the claim i is obvious.

When $|\su(f)|=2$, there are two conditions: (1) the non-zero values are not on two complementary inputs; (2) the non-zero values are on two complementary inputs. These two conditions combined with Lemma \ref{lemmaaffineeosupport} show the specific form of $f$ as in ii.

When $|\su(f)|=4$, according to Lemma \ref{lemmaaffineeosupport}, $\su(f)$ must be two pair of complementary inputs. Since the constraints are already established, there should not be any other non-trivial binary relations between the input variables. Then $f$ is reducible, and can be described as the form in iii.
\end{proof}

\begin{lemma}\label{lem:arity4inA}
Suppose $f\in \mathscr{A}$ is a non-trivial $\eo$ signature of arity 4. Then one of the following statements holds up to variable reordering:

i. $f\in\mathscr{A}_1$, and $f=\lambda\Delta_0^{\otimes2}\otimes\Delta_1^{\otimes2},\lambda\neq0$. 

ii. $f\in\mathscr{A}_2$, and either (1) or (2) holds.

(1) $f=\Delta_0\otimes\Delta_1\otimes\neq_2^{a,b}, ab\neq0, a/b\in\{\pm1,\pm\mathfrak{i}\}$.

(2) $f=(\neq_4^{a,b}), ab\neq0, a/b\in\{\pm1,\pm\mathfrak{i}\}$.

iii. $f\in\mathscr{A}_4$, and there exists distinct $\alpha,\beta$ such that  $\su(f)=\{\alpha,\overline{\alpha},\beta,\overline{\beta}\}$. Furthermore, if $x,y\in\su(f)$, then $f(x)/f(y)\in\{\pm1,\pm\mathfrak{i}\}$. In addition, $f(\alpha)f(\overline{\alpha})=\pm f(\beta)f(\overline{\beta})$.

\end{lemma}

\begin{proof}
Case i and Case ii are similar to the proof of Lemma 2 with the help of the definition of $\mathscr{A}$ and Lemma \ref{lemmaaffineeosupport}.

If $f\in\mathscr{A}_4$, then by Lemma \ref{lemmaaffineeosupport} $\su(f)$ contains two pairs of complementary inputs, namely $(\alpha,\overline{\alpha})$ and $(\beta,\overline{\beta})$. Since $f=\lambda \cdot\chi_{AX} \cdot\mathfrak i^{L_1(X)+...+L_m(X)}$, in which each $L_j$ is a linear function on field $\mathbb{F}_2$ and $X=(x_1,x_2,x_3,x_4,1)$. We can see that for $f(\alpha)f(\overline{\alpha})$ and $f(\beta)f(\overline{\beta})$, since $\alpha\oplus\overline\alpha=\beta\oplus\overline\beta=\mathbf{1}$, the difference between each $L_j(\alpha)+L_j(\overline\alpha)$ and $L_j(\beta)+L_j(\overline\beta)$ is 0 or 2. Therefore, the difference in exponents is an even integer. That is, $f(\alpha)f(\overline{\alpha})=\mathfrak i^sf(\beta)f(\overline{\beta})$ with $s$ as an even number.
\end{proof}

\begin{lemma}\label{lem:arity4inM}
Suppose $f\in \mathscr{M}\otimes\Delta_1$ is a non-trivial $\eo$ signature of arity 4. Then one of the following statements holds up to variable reordering:

i. $f\in\mathscr{M}_1\otimes\Delta_1$, and $f=\lambda\Delta_0^{\otimes2}\otimes\Delta_1^{\otimes2},\lambda\neq0$.

ii. $f\in\mathscr{M}_2\otimes\Delta_1$, and $f=\Delta_0\otimes\Delta_1\otimes\neq_2^{a,b}, ab\neq0$.

iii. $f\in\mathscr{M}_3\otimes\Delta_1$.

Signatures in $\widetilde{\mathscr{M}}\otimes\Delta_0$ can be classified similarly.
\end{lemma}

\begin{remark}\label{remarkA4-P4}
It's noticeable that $\mathscr{A}_1=\mathscr{P}_1$ and $\mathscr{A}_2\subseteq\mathscr{P}_2$, while $\mathscr{A}_4$ and $\mathscr{P}_4$ do not contain each other. If $f\in\mathscr{A}_4-\mathscr{P}_4$, then $f(\alpha)f(\overline{\alpha})=-f(\beta)f(\overline{\beta})$, where the symbols are the same as those used in the proof of Lemma \ref{lem:arity4inA}.
\end{remark}

\begin{remark}\label{remark:difference between M and M'}

It's easy to verify that $\mathscr{M}_1\otimes\Delta_1=\widetilde{\mathscr{M}_1}\otimes\Delta_0$ and $\mathscr{M}_2\otimes\Delta_1=\widetilde{\mathscr{M}_2}\otimes\Delta_0$. Besides, signatures in $\mathscr{M}_3$ and $\widetilde{\mathscr{M}_3}$ are irreducible, and consequently $\mathscr{M}_3\otimes\Delta_1\cap\widetilde{\mathscr{M}_3}\otimes\Delta_0=\emptyset$.
\end{remark}

\subsection{Proof of tractability} \label{subsection:4aritysetalgorithm}
First we will give the proof of the tractability part in Theorem \ref{arity4setdichotomy}.
\begin{proof}[Proof of tractability]
a. 
Suppose $I$ is an instance of $\hol(\neq_2\mid\mathcal{F})$ with $n$ vertices, and there are $k$ vertices assigned $f_1,...f_k\in\mathscr{M}\otimes\Delta_1$. As these signatures are reducible, we replace these signatures with $k$ signatures that belong to $\mathscr{M}$ and $k$ unary signatures $\Delta_1$ in $I$.

Suppose there is a vertex assigned $g\in\mathscr{P}$. 
 If $g$ is a binary \eo\ signature, its connection with $\Delta_1$ will produce either another $\Delta_1$ or cause the entire instance to evaluate to zero.
If $g$ is quaternary, according to Lemma \ref{lem:arity4inP}, it's easy to verify that connecting $g$ with $\Delta_1$ by $\neq_2$ has two possible results. The value of the instance becomes 0, or another $\Delta_1$ and a $\neq_2^{a,b}$ are generated ($ab\neq 0$ is not necessary here).

If $g$ is a signature in $\mathscr{M}$, connecting $g$ with a $\Delta_1$ will result in a $\neq_2^{a,b}\in \mathscr{P}$, or cause the value of the whole instance to be 0.

Therefore, regardless of the trivial cases in which the value of $I$ becomes 0, the number of $\Delta_1$ do not decrease unless it's connected to a signature in $\mathscr{M}$. So step by step, we can calculate the connection of all $\Delta_1$ with another function. Since the number of signatures is finite, the calculation will stop after finite steps. And all signatures in $\mathscr{M}$ must have been already connected by a $\Delta_1$. 

The whole calculation to eliminate $\Delta_1$ can be finished in polynomial time, since calculating one connection needs $O(1)$ time and the number of steps is smaller than $O(n)$. After the elimination, all the remaining signatures belong to $\mathscr{P}$ and consequently the obtained instance is polynomial-time computable.

b.
Similarly, given an instance $I$, we eliminate all $\Delta_1$, such that each  signature in $\mathscr{M}_\mathscr A$ is connected to a $\Delta_1$, which will also generate a generalized binary disequality in $\mathscr{A}$ according to the definition of $\mathscr{M}_\mathscr A$. In addition, gadget construction is closed under $\mathscr{A}$, so ultimately all signatures after the elimination are in $\mathscr{A}$. Again, the elimination can be done in polynomial time and the obtained instance is polynomial-time computable.

c.Similar to a.

d.Similar to b.
\end{proof}

\subsection{Proof of hardness}\label{subsection:4aritysethardness}
Firstly, we introduce several no-mixing lemmas.

\begin{lemma}[$\mathscr{A},\mathscr{P}$ no-mixing]\label{lem:mixAPhard}
Suppose $f$ and $g$ are both $\eo$ signatures with arity less than or equal to 4. Suppose $f\in\mathscr{A}-\mathscr{P}$ and $g\in\mathscr{P}-\mathscr{A}$, then $Holant(\neq_2\mid\{f,g\})$ is \#P-hard.
\end{lemma}

\begin{proof}
Since $\mathscr{A}_1\subseteq\mathscr{P}_1$ and $\mathscr{A}_2\subseteq\mathscr{P}_2$, $f\in\mathscr{A}_4-\mathscr{P}_4$. According to remark \ref{remarkA4-P4}, $\su(f)=\{\alpha,\overline{\alpha},\beta,\overline{\beta}\}$, and $f(\alpha)f(\overline{\alpha})=-f(\beta)f(\overline{\beta})$. Without loss of generality, we can assume that 

$$M(f)_{x_1x_2,x_3x_4}=\left(\begin{array}{cccc}
0 & 0 & 0 & 0\\
0 & 1 & a & 0\\
0 & b & -ab & 0\\
0 & 0 & 0 & 0
\end{array} 
\right),$$

where $a,b\in\{\pm1,\pm\mathfrak{i}\}$. $g$ is either binary or quaternary.

\begin{enumerate}

\item $g$ is a binary signature that satisfies $M(g)_{y_1,y_2}=M(\bi{c}{d})_{y_1,y_2}$, where $cd\neq 0$ and $c/d\notin\{\pm1,\pm \mathfrak{i}\}$.

We construct a gadget that connects $f$'s edge $x_1$ with $g$'s edge $y_1$ by $\neq_2$, then we get an $\eo$ signature $h$ satisfying

$$M(h)_{y_2x_2,x_3x_4}=\left(\begin{array}{cccc}
0 & 0 & 0 & 0\\
0 & d & ad & 0\\
0 & bc & -abc & 0\\
0 & 0 & 0 & 0
\end{array} 
\right).$$ 

It's easy to check that $h\notin\mathscr{A},\mathscr{P},\mathscr{M}\otimes\Delta_1,\widetilde{\mathscr{M}}\otimes\Delta_0$, consequently $\hol(\neq_2\mid h)$ is \#P-hard according to Theorem \ref{thm:sixvertexdichotomyMform}. In addition, $\hol(\neq_2\mid h)\leq_T\hol(\neq_2\mid\{f,g\})$, so the latter problem is also \#P-hard.

\item $g$ is a quaternary signature. Since it's in $\mathscr{P}-\mathscr{A}$, $|\su(g)|\geq 2$ according to Lemma \ref{lem:arity4inP} and \ref{lem:arity4inA}.
There are three possible forms of $g$. For clarity, We continue to use the category number in the description of Lemma \ref{lem:arity4inP}. 

ii.(1) $g=\Delta_0\otimes\Delta_1\otimes\neq_2^{a,b}$, $ab\neq0,a/b\notin\{\pm1,\pm \mathfrak{i}\}$. We add a self-loop using $\neq_2$ on two variables related to $\Delta_0$ and $\Delta_1$, then we get the signature $\neq_2^{a,b}$. This case then can be reduced to case 1.

ii.(2) Without loss of generality, suppose

$$M(g)_{x_1x_2,x_3x_4}=\left(\begin{array}{cccc}
0 & 0 & 0 & 0\\
0 & 0 & a & 0\\
0 & b & 0 & 0\\
0 & 0 & 0 & 0
\end{array} 
\right),$$

where $ab\neq0,a/b\notin\{\pm1,\pm\mathfrak{i}\}$. We add a self-loop on $x_1$ and $x_3$ to get $h=(\bi{b}{a})$, then it can also be reduced to case 1.

iii. Without loss of generality, suppose 

$$M(g)_{x_1x_2,x_3x_4}=\left(\begin{array}{cccc}
0 & 0 & 0 & 0\\
0 & 1 & a & 0\\
0 & b & ab & 0\\
0 & 0 & 0 & 0
\end{array} 
\right),$$

where $a\neq 0,\pm1,\pm \mathfrak{i}$. We add a self-loop on $x_1$ and $x_2$ to get $h=(\bi{1+b}{a+ab})$. If $b\neq -1$, then it can be reduced to case 1. If $b=-1$, we add a self-loop on $x_1$ and $x_3$ to get $h=(\bi{b}{a})$, then it can also be reduced to case 1.
\end{enumerate}

In summary, the problem $\hol(\neq_2\mid\{f,g\})$ is \#P-hard.
\end{proof}

\begin{lemma}[$\mathscr{M}\otimes\Delta_1,\widetilde{\mathscr{M}}\otimes\Delta_0$ no-mixing]\label{lem:mixMM'hard}
Suppose $f$ and $g$ are both quaternary $\eo$ signatures. Suppose $f\in\mathscr{M}\otimes\Delta_1-\widetilde{\mathscr{M}}\otimes\Delta_0$ and $g\in\widetilde{\mathscr{M}}\otimes\Delta_0-\mathscr{M}\otimes\Delta_1$, then $\hol(\neq_2\mid\{f,g\})$ is \#P-hard.
\end{lemma}
\begin{proof}
According to remark \ref{remark:difference between M and M'}, $f\in\mathscr{M}_3\otimes\Delta_1$ and $g\in\widetilde{\mathscr{M}_3}\otimes\Delta_0$. After normalization, we suppose $f(x_1,x_2,x_3,x_4)=[0,(1,a,b),0,0](x_1,x_2,x_3)\otimes\Delta_1(x_4)$ and $g(y_1,y_2,y_3,y_4)=[0,0,(1,c,d),0](y_1,y_2,y_3)\otimes\Delta_0(y_4)$. We connect $x_1$ to $y_1$ and $x_4$ to $y_4$, then we get

$$M(h)_{x_2x_3,y_2y_3}=\left(\begin{array}{cccc}
0 & 0 & 0 & b\\
0 & c & d & 0\\
0 & ac & ad & 0\\
0 & 0 & 0 & 0
\end{array} 
\right).$$

By Theorem \ref{thm:sixvertexdichotomyMform}, $\hol(\neq_2\mid h)$ is \#P-hard.

Since $\hol(\neq_2\mid h)\leq_T\hol(\neq_2\mid\{f,g\})$, the latter problem is also \#P-hard.
\end{proof}

\begin{lemma}[$\mathscr{A},\mathscr{M}\otimes\Delta_1$ / $\mathscr{A},\widetilde{\mathscr{M}}\otimes\Delta_0$ no-mixing]\label{lem:mixAP and Mtensorhard}
Suppose $f$ and $g$ are both $\eo$ signatures with arity less than or equal to 4. Suppose $f\in\mathscr{A}-\mathscr{P}$ and $g\in\mathscr{M}\otimes\Delta_1-\mathscr{M}_\mathscr{A}\otimes\Delta_1$, then $\hol(\neq_2\mid\{f,g\})$ is \#P-hard.

Similarly, if $f\in\mathscr{A}-\mathscr{P}$ and $g\in\widetilde{\mathscr{M}}\otimes\Delta_0-\widetilde{\mathscr{M}_\mathscr{A}}\otimes\Delta_0$, then $\hol(\neq_2\mid\{f,g\})$ is \#P-hard.
\end{lemma}
\begin{proof}

Suppose $g=[0,(a,b,c),0,0]\otimes\Delta_1$, then by connecting $\Delta_1$ to $g$ itself, it's possible to get three binary signatures $\bi{a}{b}$, $\bi{b}{c}$ and $\bi{c}{a}$. By the definition of $\mathscr{M}_\mathscr{A}$, one of these three signatures is in $\mathscr{P}-\mathscr{A}$, so with Lemma \ref{lem:mixAPhard} we have $\hol(\neq_2\mid\{f,g\})$ is \#P-hard.

The proof of the second part of the lemma is similar.
\end{proof}

With the no-mixing lemmas, we give the complete proof of the hardness for Theorem \ref{arity4setdichotomy}.

\begin{proof}[Proof of hardness]
Suppose $\mathcal{F}$ is a set of $\eo$ signatures with arity less than or equal to 4. By Theorem \ref{thm:sixvertexdichotomyMform}, if there exists $f\in\mathcal{F}$ such that $f\notin\mathscr{A}$, $\mathscr{P}$, $\mathscr{M}\otimes\Delta_1$, or $\widetilde{\mathscr{M}}\otimes\Delta_0$, then $\hol(\neq_2|\mathcal{F})$ is \#P-hard. So we assume that $\mathcal{F}\subseteq\mathscr{A}\cup\mathscr{P}\cup\mathscr{M}\otimes\Delta_1\cup\widetilde{\mathscr{M}}\otimes\Delta_0$.

We use $S_1,\widetilde{S_1},S_2,\widetilde{S_2}$ to denote $\mathscr{P}\cup\mathscr{M}\otimes\Delta_1$, $\mathscr{P}\cup\widetilde{\mathscr{M}}\otimes\Delta_0$, $\mathscr{A}\cup\mathscr{M}_{\mathscr{A}}\otimes\Delta_1$, $\mathscr{A}\cup\widetilde{\mathscr{M}_{\mathscr{A}}}\otimes\Delta_0$ respectively. With Remark \ref{remarkA4-P4} and \ref{remark:difference between M and M'}, we have that $S_1\cup S_2-\widetilde{S_1}\cup \widetilde{S_2}=\mathscr{M}_3\otimes\Delta_1-\widetilde{\mathscr{M}_3}\otimes\Delta_0$ and $\widetilde{S_1}\cup \widetilde{S_2}-S_1\cup S_2=\widetilde{\mathscr{M}_3}\otimes\Delta_0-\mathscr{M}_3\otimes\Delta_1$. If $\mathcal{F}\not\subseteq S_1\cup S_2$ and $\mathcal{F}\not\subseteq \widetilde{S_1}\cup \widetilde{S_2}$, then there exist a signature in $\mathscr{M}_3\otimes\Delta_1-\widetilde{\mathscr{M}_3}\otimes\Delta_0$ and a signature in $\widetilde{\mathscr{M}_3}\otimes\Delta_0-\mathscr{M}_3\otimes\Delta_1$. By Lemma \ref{lem:mixMM'hard}, $\hol(\neq_2|\mathcal{F})$ is \#P-hard.

Otherwise, $\mathcal{F}$ is the subset of $S_1\cup S_2$ or $\widetilde{S_1}\cup \widetilde{S_2}$. Without loss of generality, suppose $\mathcal{F}\subseteq S_1\cup S_2$. Notice that $S_1-S_2$ is equal to $(\mathscr{P}-\mathscr{A})\cup(\mathscr{M}\otimes\Delta_1-\mathscr{M}_\mathscr{A}\otimes\Delta_1)$ and $S_2-S_1=\mathscr{A}-\mathscr{P}$. If $\mathcal{F}$ contains $f$ and $g$ that belong to $S_1-S_2$ and $S_2-S_1$ respectively, then $g\in\mathscr{A}-\mathscr{P}$ and there are two cases of $f$.

\begin{enumerate}
\item $f\in\mathscr{P}-\mathscr{A}$: by Lemma \ref{lem:mixAPhard}, $\hol(\neq_2|\mathcal{F})$ is \#P-hard.

\item $f\in\mathscr{M}\otimes\Delta_1-\mathscr{M}_\mathscr{A}\otimes\Delta_1$:by Lemma \ref{lem:mixAP and Mtensorhard}, $\hol(\neq_2|\mathcal{F})$ is \#P-hard.
\end{enumerate}

In summary, $\hol(\neq_2|\mathcal{F})$ is \#P-hard unless one of $\mathcal{F}\subseteq S_1$,  $\mathcal{F}\subseteq S_2$, $\mathcal{F}\subseteq \widetilde{S_1}$ and $\mathcal{F}\subseteq \widetilde{S_2}$ holds.
\end{proof}

\section{A dichotomy for pure $\eo$ signatures} \label{section:allup}

In this section, we prove Theorem \ref{thm:puredichotomy}. Firstly, we present a crucial property of pure-up signatures.

\begin{lemma} \label{lem:alluphasdelta1}
Suppose the size of $X$ is $2d$. If an affine subspace $L$ of $\{0,1\}^X$ is contained in $\eog$, and $L \cap \hw{>}$ is not empty, then there exists $x \in X$, such that for all $\alpha \in L$, $\alpha_x=1$. 
\end{lemma}

\begin{proof}
Suppose the dimension of $L$ is $r$. We assume that $L=\beta_0+ \langle \beta_1, \beta_2, \ldots, \beta_r \rangle$, where $\beta_1, \beta_2, \ldots, \beta_r$ are linearly independent and form the basis of the associated linear space.    

Fix an arbitrary $x \in X$. 
If there is a $\beta_i$ which is 1 at $x$, then $f(\alpha)=\alpha\oplus\beta_i$ forms a 1-1 mapping. Besides, for each $\alpha\in L$, $\alpha\oplus\beta_i\in L$. Consequently, half of strings in $L$ is 0 at $x$, and the other half is 1 at $x$.
If there is no $\beta_i$ which is 1 at $x$, then all strings in $L$ is equal to $\beta_0$ at the $x$-th bit. We call it as a fixed $0$ or $1$ at $x$.

There are $2^r$ strings in $L$, and $2d \cdot 2^r$ bits in these strings in total. Since $L$ is contained in $\eog$, and $L \cap \hw{>}$ is not empty, there are strictly more 1's than 0's. 
Hence, considering all $x\in X$, there must be strictly more fixed $1$'s than fixed 0's. That is, there is at least one $x\in X$, such that for all $\alpha\in L$, $\alpha_x=1$. 
\end{proof}

Pure-up signatures also have good closure property under specific operations, as stated in the lemma below.

\begin{lemma}  \label{lem: allup closed under pin}
    Pure-up signatures are closed under the operation of adding a self-loop by the pinning signature $\bi{0}{1}$, which is also denoted as $\pin$. 
\end{lemma}
\begin{proof}
    Suppose $f$ is a pure-up $\eo$ signature. As $\text{Span}(f)\subseteq\eog$, for every odd number of strings in $\su(f)$, their summation has no less 1's than 0's.

    Now we add a self-loop by the pinning signature $\pin$ on two variables of $f$, say $x_1$ and $x_2$, to get $f'$. If $\pin$ enforces some fixed variables of $f$ to be the opposite value, then $f'\equiv0$, which is a trivial pure-up signature.

   Otherwise without loss of generality we assume that $\pin$ fix $x_1$ to 0 and $x_2$ to 1. For each $\alpha\in\su(f')$, $01\alpha\in\su(f)$. Then for any odd integer $k\geq3$ and $\alpha_1,\cdots,\alpha_k\in\su(f')$, we have  $01\alpha_1,\cdots,01\alpha_k\in\su(f)$, and consequently $01\alpha_1\oplus\cdots\oplus01\alpha_k=01\alpha$. By the pure-up property of $f$ we have $01\alpha$ has no less 1's than 0'. Therefore $\alpha$ has no less 1's than 0' and $f'$ is also pure-up.
\end{proof}

There are two cases of a pure-up signature $f$. The first case is that there exists a string in $\text{Span}(f)$ having strictly more 1's than 0's. By Lemma \ref{lem:alluphasdelta1} we know that there is at least one variable of $f$ fixed to 1. The second case is that all strings in $\text{Span}(f)$ are in $\eoe$. By Lemma \ref{lemmaaffineeosupport} we know that $f$ is $\eom$. 

Given this intuition, we introduce an polynomial time approach to transform $\#\eo(\mathcal{F})$ to a $\#\csp$ problem, where $\mathcal{F}$ contains only pure-up signatures. We use the solution space of a signature grid $\Omega$, to denote the set of assignments on variables that give non-zero evaluations in the instance.

\begin{lemma} \label{lem:alluptocsp}
    Suppose $\mathcal{F}$ is composed of pure-up signatures. Then for any instance $\Omega(G(U,V,E),\pi)$ of $\hol(\neq_2\mid\mathcal{F})$, in polynomial time we can obtain another instance $\Omega'(G(U,V,E),\pi')$ of $\hol(\neq_2\mid\mathcal{F}')$ such that $\hol_{\Omega'}=\hol_\Omega$ , where $\mathcal{F}'=\{f|_{\eom[P]}\mid f\in\mathcal{F},\text{ P is a pairing of Var}(f)\}$.

\end{lemma}

\begin{proof}

Given an instance $I$ of $\hol(\neq\mid\mathcal{F})$ with $n$ signatures, since each signature $f\in \mathcal{F}$ is pure-up, it either satisfies conditions in Lemma \ref{lemmaaffineeosupport}   or satisfies conditions in Lemma \ref{lem:alluphasdelta1} , depending on whether the intersection with $\hw{>}$ is empty or not. 

If $f$ satisfies conditions in Lemma \ref{lemmaaffineeosupport} , we get $\su{(f)} \subseteq \eom[P]$ for some $P$ as a pairing of $\text{Var}(f)$. We remark that for this kind of signatures, if we fix an arbitrary variable to be 0 (or 1), there will be another variable fixed to 1 (or 0).

Suppose the signature grid of $I$ is $\Omega(G(U,V,E),\pi)$, and $\pi$ assigns $\neq_2$ to all vertices in $U$. If for each $v\in V$, $f_v$ assigned by $\pi$ satisfies the conditions in Lemma \ref{lemmaaffineeosupport}, then the instance already meets our requirements of $\Omega'$ in the description of this lemma. We are done by setting $\Omega'=\Omega$. We denote this case as the $\eom$ case.

Otherwise, there exist a signature $f$ that satisfies the conditions in Lemma \ref{lem:alluphasdelta1}. Let $x_1\in\text{Var}(f)$ be the variable fixed to 1. Consider the $\neq_2$ signature related with $x_1$. If the other end is connected to an $\eom$ signature $h\in\mathcal{F}$, then it will force one of its variables to be 0. By the analysis above, another variable of $h$ is fixed to 1. This mechanism can be triggered successively, until the pinning edge is connected to a non-$\eom$ signature $g\in\mathcal{F}$, and $y_2\in\text{Var}(g)$ is fixed to 0. 

As $g$ is a non-$\eom$ signature $g$, a variable $y_1$ of $g$ is fixed to 1. If $y_1$ and $y_2$ are the same variable, the partition function of $I$ would be 0 and consequently trivial, so we may assume $y_1$ and $y_2$ are distinct. 

After $y_1$ and $y_2$ are fixed to 1 and 0 respectively, by Lemma \ref{lem: allup closed under pin}, $g^{y_1=1,y_2=0}$ is still pure-up. We replace $g$ with $g'=\Delta_1\otimes\Delta_0\otimes g^{y_1=1,y_2=0}$ in the instance. The space of $\text{Span}(g)$ in $I$ is restricted by the pair $(y_1,y_2)$ under this operation, while the partition function keeps unchanged. This operation can be repeated as long as there exists a signature that is not $\eom$ in the instance. Again we are done by the $\eom$ case. Since the instance has $n$ signatures, and each operation needs $O(n)$ time, the process will terminate in $O(n^2)$ time.
\end{proof}

Along with Theorem \ref{thm:CSPdichotomy} and Lemma \ref{thm:csp=eom}, we prove the tractability part of Theorem \ref{thm:puredichotomy}.

\begin{proof}[Proof of tractability]
Given an instance of $\hol(\neq_2\mid\mathcal{F})$, by Lemma \ref{lem:alluptocsp} we can transform it to a instance of $\hol(\neq_2\mid\mathcal{F}')$ in polynomial time, where $\mathcal{F}'=\{f|_{\eom[P]}\mid f\in\mathcal{F},\text{ P is a pairing of Var}(f)\}$. By Definition \ref{def:eoaeop}, $\mathcal{F}'\subseteq\mathscr{A}$ or $\mathcal{F}'\subseteq\mathscr{P}$. Then by Lemma \ref{thm:csp=eom} and the polynomial time algorithms in Theorem \ref{thm:CSPdichotomy} for $\#\csp$ problems, the partition function of the latter instance can be calculated in polynomial time.
\end{proof}

Next we prove the \#P-hardness part of Theorem \ref{thm:puredichotomy}. 
Given an pure-up signature $f$, we use $\mathcal{G}(f)$ to denote the set of $\eom$ signatures that $f$ may generate. If $f$ is $\eom$, then $\mathcal{G}(f)=\{f\}$. Otherwise, $f$ has at least one variable fixed to be 1. We add a self-loop on this variable and another arbitrary variable by $\neq_2$ to get $f'$. If $f'$ is $\eom$, then $f'\in\mathcal{G}(f)$, otherwise we repeat this process on $f'$. $\mathcal{G}(f)$ is defined by this recursive process. 

\begin{lemma}\label{lem:alluphardness}
Suppose $\mathcal{F}$ is a finite set of signatures composed of pure-up  signatures. If $\bigcup_{f\in\mathcal{F}}\mathcal{G}(f)$ is not a subset of $\mathscr{A}$ or $\mathscr{P}$, then $\#\eo(\mathcal{F})$ is \#P-hard.
\end{lemma}
\begin{proof}
    All signatures in $\bigcup_{f\in\mathcal{F}}\mathcal{G}(f)$ are generated by signatures in $\mathcal{F}$ through gadget construction, thus we have $\#\eo(\bigcup_{f\in\mathcal{F}}\mathcal{G}(f))\leq_T\#\eo(\mathcal{F})$. By Lemma \ref{thm:csp=eom}, $\#\eo(\bigcup_{f\in\mathcal{F}}\mathcal{G}(f))$ is equivalent to a $\#\csp(\mathcal{F}')$ problem, where $\mathcal{F}'$ is not a subset of $\mathscr{A}$ or $\mathscr{P}$. By Theorem \ref{thm:CSPdichotomy} $\#\csp(\mathcal{F}')$ is \#P-hard, and consequently $\#\eo(\mathcal{F})$ is \#P-hard.
\end{proof}

\begin{proof}[Proof of hardness]
    Assume $\mathcal{G}(f)\subseteq\mathscr{A}$ (or $\mathscr{P}$ respectively) is of arity $2d$. For each pairing $P$ of $\text{Var}(f)$, we add self-loop by $\neq_2$ on each variable fixed to 1 and its partner in the pair until we get an $\eom$ signature $f'\in\mathcal{G}(f)$ of arity $2k$. $f|_{\eom[P]}=f'\otimes\Delta_0^{2d-2k}\otimes\Delta_1^{2d-2k}$, therefore  $f|_{\eom[P]}\in\mathscr{A}$ (or $\mathscr{P}$ respectively). This shows $f$ is $\eom[\mathscr{A}]$ (or $\eom[\mathscr{P}]$ respectively).

    Consequently, if $\mathcal{F}$ is not a subset of $\eom[\mathscr{A}]$ or $\eom[\mathscr{P}]$, then $\bigcup_{f\in\mathcal{F}}\mathcal{G}(f)$ is not a subset of $\mathscr{A}$ or $\mathscr{P}$, and by Lemma \ref{lem:alluphardness} $\#\eo(\mathcal{F})$ is \#P-hard. 
\end{proof}

\section{Polynomial time algorithm for specific 0-rebalancing and 1-rebalancing signatures}\label{balancing}

In this section, we prove Theorem \ref{thm:rebaalg}. The mapping $\psi$ in Definition \ref{def:reba} particularly denotes the correspondence in the first level of the recursive definition, and we call it the first level mapping of $f$. 
Because the arity $2d-2$ signature $f^{x=0,\psi(x)=1}$ is also \ba[0], we can apply the definition again to $f^{x=0,\psi(x)=1}$, and get the first level mapping of $f^{x=0,\psi(x)=1}$. From the perspective of $f$, this mapping is a second level mapping, and we denote it by $\psi_2^{x=0,\psi(x)=1}$. 
Since \ba[0] is defined in an adaptive recursive way, the second level mappings depend on which $x$ is fixed to $0$. 
This rebalancing process can continue until the arity comes to 0, which implies that $f$ can always, sends out a $1$ after receiving a $0$. 

We now prove that the set of \ba[0] signatures is closed under gadget construction.

\begin{lemma} \label{lem:ba-tensor-closed}
If $f$ and $g$ are \ba[0], then $f\otimes g$ is also \ba[0].
\end{lemma}

\begin{proof}
We will prove this lemma by induction on the sum of arity of $f$ and $g$. It's obvious that when the sum is no more than 4, both of them must be binary signatures, thus trivially correct since no matter which variable $x$ of $f\otimes g$ is chosen, after remove $x$ and $\psi(x)$, there is only $f$ or $g$ left, which is also \ba[0]. The mapping $\psi$ is exactly the first level mapping of the signature that $x$ is related to.

Assume that the lemma is already proved for $f$ and $g$, where the sum of arity of $f$ and $g$ is no more than $2d$. Now we consider the situation that the sum is more than $2d$.

Suppose \ba[0] signatures $f$ and $g$ have arity $a$ and $b$ respectively, and $a+b=2d+2$. We assume that $f$ and $g$ are not 0 signatures, therefore $f\otimes g\not\equiv 0$.
Let $\psi_f$ and $\psi_g$ denote the first level mapping of $f$ and $g$, which witness the \ba[0] property of $f$ and $g$ respectively. 
We define that $\psi_{f\otimes g} = \psi_f \cup \psi_g$. 
Suppose $x$ is an arbitrary variable of $f \otimes g$. If $x$ is related to $f$, then $[x =0 \to \psi_{f\otimes g}(x)= \psi_f(x) =1]$. Hence, $f \otimes g \not\equiv 0 \to  [\forall x, x =0 \to \psi_{f\otimes g}(x) =1]$. If we fix $x$ and $\psi(x)$, we get $(f \otimes g)^{x=0,\psi_{f\otimes g}(x) =1} $, which is in fact $f^{x=0,\psi_{f}(x) =1} \otimes g$, and it is \ba[0] according to the induction hypothesis. If $x$ is related to $g$, the proof is similar. Therefore, the induction is done and the lemma is proved. 
\end{proof}

\begin{lemma} \label{lem:ba-jumper-closed}
Suppose $f$ is a \ba[0] signature, $g$ is obtained from $f$ by adding an arbitrary self-loop by $\neq_2$. Then $g$ is also \ba[0].
\end{lemma}

\begin{proof}
We use induction to prove this lemma. If the arity of $f$ is 2, the conclusion holds naturally since all 0 arity signatures are \ba[0].
If the arity of $f$ is 4, the conclusion holds simply because $g$ is a $\eo$ signature of arity 2, that is a general disequality. 

Assume that the conclusion holds for $f$ of arity no more than $2d$. 

Suppose $f$ is a \ba[0] signature of arity $2d+2$, with variable set $\text{Var}(f)=\{s,t,x_1, \cdots, x_{2d}\}$.
Without loss of generality, let $g(x)=f(01x)+f(10x)$, whose variable set is $X=\{x_1, \cdots, x_{2d}\}$. 

Suppose the first level mapping of $f$ is $\psi$. For any variable $x \in X$,  $y=\psi(x) \in X \cup \{s,t\} -\{x\}$. Define
\[
\psi'(x)= \left \{ \begin{array}{ll}
    \psi(x) &  \text{ if } x\neq s,t,  \psi(x)\in X; \\
    \psi_2^{x=0,s=1}(t) &    \text{ if }  \psi(x)=s ;\\
    \psi_2^{x=0,t=1}(s) &    \text{ if } \psi(x)=t .
\end{array}  \right.
\]

We use $\psi'$ to witness that $g$ is also \ba[0], and there are several cases. 
  
\begin{itemize}
    \item  $y \in X$. 

For any $\alpha \in \{0,1\}^X$, if $\alpha_x=\alpha_y=0$, then $f(01\alpha)=f(10\alpha)=0$ because $y=\psi(x)$, 
so $g(\alpha)=0$. 

Because $f ^{x=0,y=1}$ is \ba[0], by the induction hypothesis, $f ^{x=0,y=1, s\neq t}$ is \ba[0]. That is, $g^{x=0, \psi'(x)=1}$ is also \ba[0].
    
    \item   $y=s$. 

For any $\alpha \in \{0,1\}^X$ with $\alpha_x=\alpha_{\psi'(x)}=0$. First we prove that $g(\alpha)=0$. Since $g(\alpha)=f(01\alpha)+f(10\alpha)$, it suffices to prove that $f(01\alpha)=0$ and $f(10\alpha)=0$.  

For $01\alpha$, we have $s=0,t=1, x=0, \psi'(x)=\psi_2^{x=0,s=1}(t)=0$. 
Because $\psi(x)=s$, that is, $x$ and $s$ can not both be 0, thus we have $f(01\alpha)=0$. 

For $10\alpha$, we have $s=1,t=0,x=0, \psi'(x)=\psi_2^{x=0,s=1}(t)=0$. 
 Obviously, $f(10\alpha)=f^{x=0,s=1}(\beta)$, where $\beta$ denotes the string by removing the first bit and the $\alpha_x$ bit from $10\alpha$. Because $f^{x=0,s=1}$ is \ba[0], with $\psi_2^{x=0,s=1}$ as the second level mapping, and $\beta$ satisfies $t=0, \psi_2^{x=0,s=1}(t)=0$, we have $f^{x=0,s=1} (\beta)=0$. That is, $f(10\alpha)=0$. 

For the second part, we need to prove $g^{x=0, \psi'(x)=1}$ is \ba[0]. $g^{x=0, \psi'(x)=1}(\beta)=f^{x=0, \psi'(x)=1}(01\beta)+f^{x=0, \psi'(x)=1}(10\beta)=f^{x=0, \psi'(x)=1,s=0,t=1}(\beta)+f^{x=0, \psi'(x)=1,s=1,t=0}(\beta)$.
Because $s= \psi(x)$, we know the first item $f^{x=0, \psi'(x)=1,s=0,t=1}$ 
is 0 signature. The second item is in fact $f^{x=0, s=\psi(x)=1,t=0,\psi_2^{x=0,s=1}(t) =1}$, which is \ba[0] by the recursive definition of $f$ in the second level. 
    \item  
    The case $y=t$ is similar as $y=s$. 
\end{itemize}

In summary, we have proved that $g$ is \ba[0], witnessed by $\psi'$, which finishes our induction. 
\end{proof}

\begin{lemma}
    Given any $\eo$ gadget composed of \ba[0] signatures, the signature of this gadget is also \ba[0].
    \label{lem:bagad}
\end{lemma}

\begin{proof}
Any gadget construction can be represented by a series of operations of two forms.

The first operation is tensor production. It is a binary operation turn $f$ and $g$ into $f \otimes g$. 
The second operation is adding a self-loop. Suppose $f$ is \ba[0] signature of arity $2d$. For any two variables $s,t$ of $f$, the self-loop operation gives an arity $2d-2$ signature $f^{s \neq t}$.  

We have proved that \ba[0] is preserved by two basic operations in Lemma \ref{lem:ba-tensor-closed} and  \ref{lem:ba-jumper-closed}, hence the conclusion holds.  
\end{proof}

It is interesting that when we observe the local signature from a global perspective of the whole instance, it shows an opposite property which is denoted as outside-\ba[1].

\begin{lemma}[outside-1-rebalancing]\label{outside-1-rebalancing}
    Suppose $\mathcal{F}$ is composed of \ba[0] signatures, and there is an instance $G(U,V,E)$ of $\hol(\neq_2\mid\mathcal{F})$. Fix an arbitrary vertex $s\in V$ assigned $f_s$, and let $E(s)=X=\{x_1,x_2, \cdots, x_{2d}\}$. We delete the vertex $s$ from $G$ to get a gadget $H$ whose dangling edges are associated to $X$. Then the signature $f_H$ of $H$ is \ba[1], and we denote this property as the outside-1-rebalancing property of $s$ on $X$ deduced by the instance $G$.
\end{lemma}

\begin{proof}
For each $1\le i\le 2d$, we denote the vertex adjacent to $x_i$ in $U$ as $u_i$, and the other neighbouring edge of $u_i$ is $x'_i$. We let $U'=\{u_1,u_2, \cdots, u_{2d}\}$ and $X'=\{x'_1,x'_2, \cdots, x'_{2d}\}$.

We delete $s$, $X$ and $U'$ from $G$ to get a gadget $H'$ whose dangling edges are associated to $X'$. Each $v \in V$ is assigned a 0-rebalancing signature, so by Lemma \ref{lem:bagad} the signature $f_{H'}$ of $H'$ is \ba[0]. 
For each $1\le i\le 2d$, since $x_i$ always takes the opposite value of $x'_i$ by the restriction from $u_i$, we have $f_H$ is \ba[1].
\end{proof}

We denote the signature of $f_H$ in Lemma \ref{outside-1-rebalancing} as $\hat{f_s}$, or simply $\hat{f}$ without causing ambiguity.
We remark that for a given vertex $s$ assigned $f$, both the 0-rebalancing property of $f$ and the 1-rebalancing property of $\hat{f}$ are associated with $E(s)$.

\begin{figure}
            \centering
            \includegraphics[height=0.2\textheight]{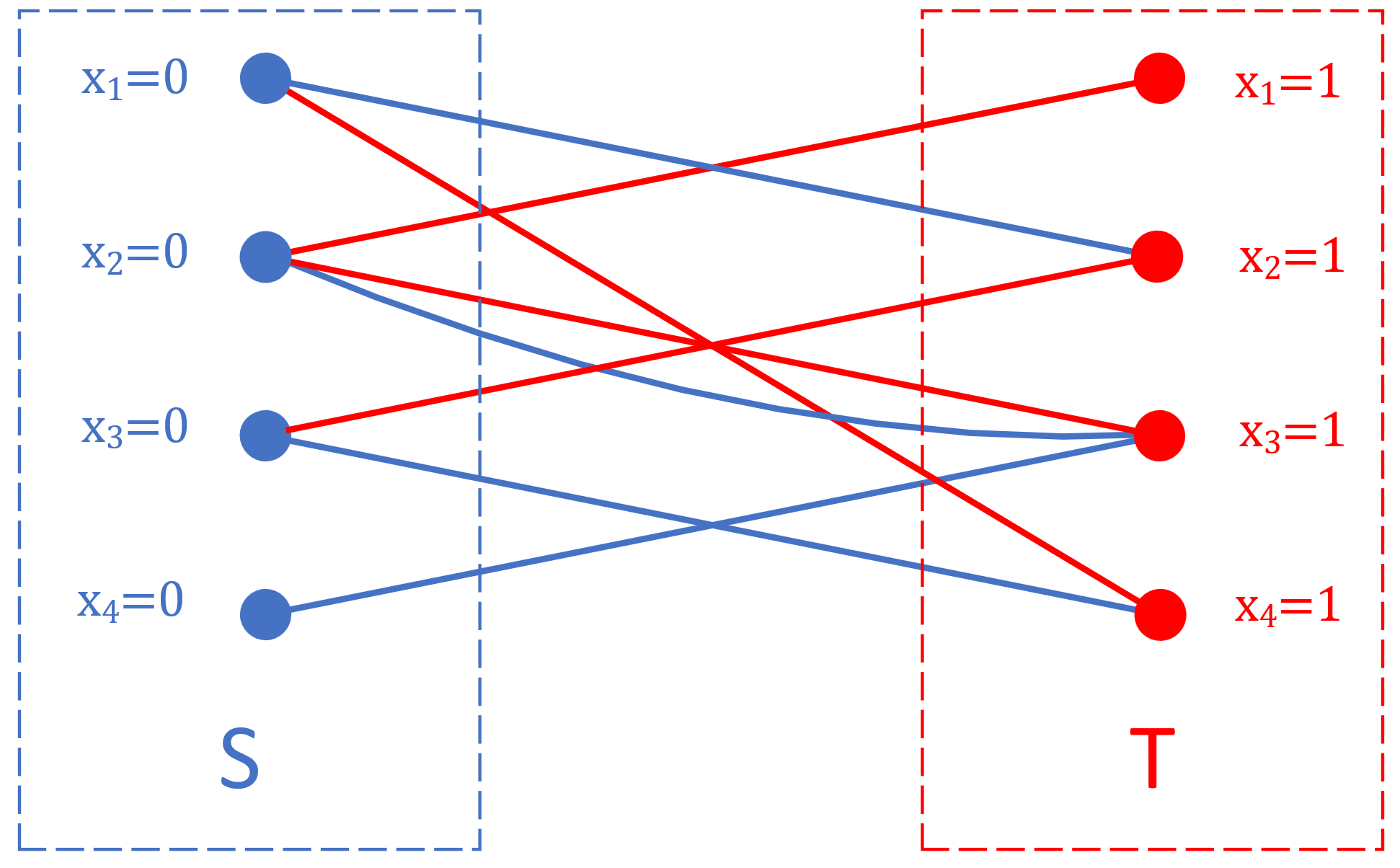}
            \caption{An example of $G_X$. Blue edges refer to directed edges in $B$, heading from $S$ to $T$. Red edges refer to directed edges in $R$, heading from $T$ to $S$.}
            \label{fig:bip}
        \end{figure}

With the two property, we again introduce an polynomial time approach to transform $\#\eo(\mathcal{F})$ to a $\#\csp$ problem, where $\mathcal{F}$ contains only \ba[0] signatures.

\begin{lemma}\label{lem:0rebtocsp}
Suppose $\mathcal{F}$ is composed of \ba[0] signatures. Then for any instance $\Omega(G(U,V,E),\pi)$ of $\hol(\neq_2\mid\mathcal{F})$, in polynomial time we can obtain another instance $\Omega'(G(U,V,E),\pi')$ of $\hol(\neq_2\mid\mathcal{F}')$ such that $\hol_{\Omega'}=\hol_\Omega$ , where $\mathcal{F}'=\{f|_{\eom[P]}\mid f\in\mathcal{F},\text{ P is a pairing of Var}(f)\}$. 
\end{lemma}

\begin{proof}
    Given an instance $I$ of $\hol(\neq_2\mid\mathcal{F})$ with $n$ signatures, let $s\in V$ assigned $f\in\mathcal{F}$ be an arbitrary vertex. By Lemma \ref{outside-1-rebalancing}, $f$ is \ba[0] and $\hat{f}$ is \ba[1]. We use $\psi$ to denote the first level mapping of $f$ and $\theta$ to to denote the first level mapping of $\hat{f}$, and both mappings are associated with $E(s)=X=\{x_1,x_2, \cdots, x_{2d}\}$.

    We use $G_X=(S,T,B\cup R)$ to denote the directed bipartite graph defined by $X,\psi$ and $\theta$ where $S=\{s_i\mid 1\le i\le 2d\},T=\{t_j\mid 1\le j\le 2d\}, B=\{s_it_j\mid 1\le i\le 2d,1\le j\le 2d, \psi(x_i)=x_j\}$ and $R=\{t_js_i\mid 1\le i\le 2d,1\le j\le 2d, \theta(x_j)=x_i\}$. See Figure \ref{fig:bip} for an example.

We remark that for each $x_i\in X$, $s_i$ and $t_i$ represent the statements $[x_i=0]$ and $[x_i=1]$ respectively. Each directed edge $s_it_j$ represents that the statement $[x_i=0\to x_j=1]$ is always True. We write $x_i=0\to x_j=1$ for short. This is guaranteed by the 0-rebalancing property. Similarly, each directed edge $t_js_i$ represents that $ x_j=1\to x_i=0$.

     As $\psi$ and $\theta$ are two mappings on $X$, the out-degree of each vertex in $S\cup T$ is exactly one. Consequently, there exists a directed cycle, namely $s_it_js_k\cdots s_i$.

     If $x_i=0$, by the edge $s_it_j$ we have $x_j=1$. If $x_j=1$, by the path $t_js_k\cdots s_i$ we have $x_i=0$. Consequently, $x_i$ and $x_j$ can not take the same value in the instance. Let $x_i=p,x_j=q$.  We add the pair $(p,q)$ to the pairing $P_s$ and replace $f$ and $\hat{f}$ with $f|_{p\neq q}$ and $\hat{f}|_{p\neq q}$.

    After this, we update $\psi$ to $\psi'$ in the following way:

     \[
\psi'(x)= \left \{ \begin{array}{ll}
q &  \text{ if }  x=p;  \\
p &  \text{ if }  x=q; \\
    \psi(x) &  \text{ if } x\neq p,q, \psi(x)\in X ; \\
    \psi_2^{x=0,p=1}(q) &   \text{ if }  x\neq p,q,\psi(x)=p ;\\
    \psi_2^{x=0,q=1}(p) &   \text{ if }  x\neq p,q,\psi(x)=q .
\end{array}  \right.
\]
Similar to the proof of Lemma \ref{lem:ba-jumper-closed}, $\psi'$ can also be a witness for the 0-rebalancing property of the signature $f|_{p\neq q}$. We also update $\theta$ to $\theta'$ similarly. 

Now consider the directed bipartite graph $G_X'$ defined by $X,\psi'$ and $\theta'$. We can see that $s_it_j,s_jt_i,t_is_j,t_js_i\in G_X'$, forming two isolated connecting components. Consequently we may ignore them and apply the analysis above to $G_X'-s_i,s_j,t_i,t_j$ to obtain another opposite pair of variables, which means this process can be done successively. At the end of the procedure, we obtain $P_v$ which is the desired pairing that restrict $f$.

At the beginning of the algorithm, $\theta$ can be computed in $O(2d\cdot n)$ time, by using the 0-rebalancing property successively on the gadget $G-s$ to compute a tour from each variable of $s$ to another variable of it. After that, finding a directed loop in $G_X$ and updating $\psi$ need $poly(d)$ time, which is a constant. Updating $\theta$, however, needs $O(n)$ time. Each vertex $s$ need $d$ steps to clarify $P_s$, and there are $n$ vertices in total. Consequently, the time complexity of the algorithm is $O(n^2)$, which is polynomial.
\end{proof}

Along with Theorem \ref{thm:CSPdichotomy} and Lemma \ref{thm:csp=eom}, we present the algorithm in Theorem \ref{thm:rebaalg}.

\begin{proof}[Proof of tractability]
Given an instance of $\hol(\neq_2\mid\mathcal{F})$, by Lemma \ref{lem:0rebtocsp} we can transform it to a instance of $\hol(\neq_2\mid\mathcal{F}')$ in polynomial time, where $\mathcal{F}'=\{f|_{\eom[P]}\mid f\in\mathcal{F},\text{ P is a pairing of Var}(f)\}$. By Definition \ref{def:eoaeop}, $\mathcal{F}'\subseteq\mathscr{A}$ or $\mathcal{F}'\subseteq\mathscr{P}$. Then by Lemma \ref{thm:csp=eom} and the polynomial time algorithms in Theorem \ref{thm:CSPdichotomy} for $\#\csp$ problems, the partition function of the latter instance can be calculated in polynomial time.
\end{proof}
\section{Conclusion}\label{sec:ccls}

In this article, we prove two dichotomies for $\#\eo(\mathcal{F})$: one with $\mathcal{F}$ restricted to binary and quaternary signatures, and another with $\mathcal{F}$ restricted to pure signatures. We also present an algorithm for rebalancing signature sets satisfying $\eom[\mathscr{A}]$ or $\eom[\mathscr{P}]$. A more detailed characterization of pure signatures and rebalancing signatures would be valuable, similar to what has been achieved for $\delta$-affine kernels in \cite{shao2024eulerian}. Indeed, pure signatures clearly exhibit close connections to $\delta$-affine signatures.

The pursuit of a complete dichotomy for $\#\eo$ remains an important research direction. Notably, the signature $f_{56}$ defined in Section~\ref{4056} is neither pure nor rebalancing. The complexity of $\#\eo(\{f_{56}\})$ remains unresolved and presents significant interest. While some progress has been made, substantial challenges persist in addressing this problem.

Recent work \cite{meng2025fpnp} has established an $\text{FP}^\text{NP}$ versus \#P dichotomy for complex-valued $\#\eo$, which includes an interesting polynomial-time algorithm with a specific NP oracle for solving $\#\mathrm{EO}(f_{56})$. This development has also prompted renewed investigation of Boolean constraint satisfaction problems \cite{feder2006classification}. Nevertheless, the polynomial-time computability of $\#\eo(f_{56})$ remains an open question, and the most general known polynomial-time algorithm for $\#\eo$ is still our rebalancing signature algorithm presented in this work.

%%
%% Bibliography
%%

%% Please use bibtex, 

\bibliography{P-time-EO}

\appendix
\end{document}